\documentclass[sigconf]{acmart}
\renewcommand\footnotetextcopyrightpermission[1]{} 

\usepackage{algorithm}
\usepackage{algorithmic}

\usepackage{subfigure}
\usepackage{xcolor}
\usepackage{multirow}
\theoremstyle{plain}
\newtheorem{theorem}{Theorem}[section]
\newtheorem{question}{Question}[section]

\newtheorem{lemma}[theorem]{Lemma}
\newtheorem{corollary}[theorem]{Corollary}
\theoremstyle{definition}
\newtheorem{definition}[theorem]{Definition}

\theoremstyle{remark}

\AtBeginDocument{%
  }

\begin{document}
\pagestyle{plain}

\title{Distributed Differentially Private Data Analytics \\ via Secure Sketching}

\author{Jakob Burkhardt, Hannah Keller, Claudio Orlandi, Chris Schwiegelshohn}
\email{{jakob, hkeller, orlandi, schwiegelshohn}@cs.au.dk}
\affiliation{%
  \institution{Aarhus University}
  \country{Denmark}
}

\newcommand{\algofont}[1]{\mathtt{#1}}

\newcommand{\mR}{\mathbb{R}}
\newcommand{\mA}{\mathbf{A}}
\newcommand{\mS}{\mathbf{S}}

\newcommand{\data}{\mathbf{a}}
\newcommand{\mdata}{\mathbf{A}}
\newcommand{\datamatrix}{\mathbf{A}}
\newcommand{\chernoff}{\gamma}

\newcommand{\datasize}{n}
\newcommand{\sketchsize}{m}
\newcommand{\predictor}{\mathbf{b}}
\newcommand{\x}{\mathbf{x}}
\newcommand{\mx}{\mathbf{X}}

\newcommand{\xopt}{\x_{\text{OPT}}}
\newcommand{\mxopt}{\mx_{\text{OPT}}}

\newcommand{\xest}{\x'}
\newcommand{\mxest}{\mx'}
\newcommand{\xn}{\x_n}
\newcommand{\xs}{\x_s}
\newcommand{\sketch}{\mathbf{S}}
\newcommand{\noise}{\mathbf{g}}
\newcommand{\mnoise}{\mathbf{G}}
\newcommand{\X}{\mathbf{X}}

\newcommand{\tr}{\algofont{trace}}
\newcommand{\magic}{\alpha}
\newcommand{\magict}{\alpha'}

\newcommand{\usvd}{\mathbf{U}}
\newcommand{\sig}{\mathbf{\Sigma}}
\newcommand{\vsvd}{\mathbf{V}}

\newcommand{\jbound}{\sigma\sqrt{360(7d + \ln(1/\Bar{\delta}))} }
\newcommand{\jboundsq}{\sigma^2\cdot360(7d + \ln(1/\Bar{\delta}))}

\newcommand{\newboundb}{\sigma\sqrt{\log(1/\beta)}}
\newcommand{\newbound}{\sigma d\sqrt{\log(1/\beta)}}
\newcommand{\newboundsq}{\sigma^2d^2\log(1/\beta)}
\newcommand{\jboundshortsq}{\sigma^2dn\log(1/\beta)}
\newcommand{\jboundshortsqb}{\sigma^2n\log(1/\beta)}

\newcommand{\const}{c}

\newcommand{\smallthing}{\omega}
\newcommand{\corruptc}{C_{cor}}
\newcommand{\corrupts}{S_{cor}}

\newcommand{\boundnon}{\sigma^2d\log(1/\Tilde{\delta})}

\newcommand{\multerr}{\alpha}
\newcommand{\adderr}{\beta}
\newcommand{\pfail}{t}

\newcommand{\vect}{x}
\newcommand{\ind}{i}
\newcommand{\entry}{\ensuremath{e_\ind}}

\newcommand{\threshold}{\alpha}
\newcommand{\spacereq}{t}
\newcommand{\successprob}{\beta}
\newcommand{\spacefunc}{f}

\newcommand{\lnorm}[1]{\lVert #1 \rVert_2}
\newcommand{\lnormbig}[1]{\Bigl\lVert #1 \Bigr\rVert_2}
\newcommand{\innerprod}[2]{\langle #1 , #2 \rangle}

\newcommand{\hnote}[1]{{\color{purple}{\sf (Hannah's Note:} {\sl{#1}} {\sf EON)}}}
\newcommand{\hannah}[1]{{\color{black} {{#1}}}}
\newcommand{\neurips}[1]{{\color{purple} {{#1}}}}

\newcommand{\squ}[1]{{\color{purple} {\sl{#1}}}}

\newcommand{\chris}[1]{{\color{teal} {Chris: \sl{#1}}}}
\newcommand{\HAN}[1]{{\color{purple} {Hannah: \sl{#1}}}}

\newcommand{\changeforclaudio}[1]{{\color{olive} {Changes for Claudio to check: \sl{#1}}}}

\newcommand{\cnote}[1]{{\color{red}{\sf (Claudio's Note:} {\sl{#1}} {\sf EON)}}}
\newcommand{\claudio}[1]{{\color{red} {\sl{#1}}}}
\newcommand{\mypar}[1]{\vspace{3pt}\noindent\textbf{#1}}

\newcommand{\jnote}[1]{{\color{blue}{\sf (Jakob's Note:} {\sl{#1}} {\sf EON)}}}

\newcommand{\N}{{\mathbb{N}}}
\newcommand{\Z}{{\mathbb{Z}}}
\newcommand{\R}{{\mathbb{R}}}
\newcommand{\E}{{\mathbb{E}}}

\newcommand{\decom}{s}

\newcommand{\sadditive}{\beta_\mS}
\newcommand{\smult}{\ensuremath{\alpha_\mS}}

\newcommand{\ltm}{linear-transformation model}

\begin{abstract}

We introduce the \emph{linear-transformation model}, a distributed model of differentially private data analysis. Clients have access to a 
trusted platform capable of applying a public matrix to their inputs. Such 
computations can be securely distributed across multiple servers using simple and 
efficient secure multiparty computation techniques.

The linear-transformation model serves as an intermediate model between the highly 
expressive \emph{central model} and the minimal \emph{local model}. In the central model, 
clients have access to a trusted platform capable of applying any function to their inputs. 
However, this expressiveness comes at a cost, as it is often prohibitively expensive to distribute such 
computations, leading to the central model typically being implemented by a single trusted 
server. In contrast, the local model assumes no trusted platform, which forces clients 
to add significant noise to their data. The linear-transformation model avoids the single 
point of failure for privacy present in the central model, while also mitigating the high 
noise required in the local model.

We demonstrate that linear transformations are very useful for differential 
privacy, allowing for the computation of linear sketches of input data. These sketches 
largely preserve utility for tasks such as private 
low-rank approximation and private ridge regression, while introducing only minimal error, critically independent 
of the number of clients.

\end{abstract}

\maketitle

\section{Introduction}
\renewcommand{\epsilon}{\varepsilon}
Differential Privacy (DP)~\citep{DMNS06} has become the de-facto standard for 
ensuring the privacy of individuals whose data is used in data analytics. DP offers strong, provable guarantees, such as composability 
and resilience to auxiliary information. In the classic \emph{central model} of differential 
privacy, clients submit their data to a 
trusted central server, which processes the data and releases a (necessarily) 
noisy version of the result. Analysts can 
then use this result to perform queries on the data. Informally, 
differential privacy guarantees that the analyst cannot 
confidently determine whether any individual contributed their data, except with some small 
probability.

The central model requires trust that the server will not reveal any individual's 
data and cannot be compromised by external actors. However, universally trusted 
servers may not exist in practice. In situations where privacy is critical but 
trusting a single server is not feasible, alternative approaches are available: One such approach is to use secure multiparty computation (MPC) 
techniques to distribute the central server's computations across multiple 
servers. Under the assumption that a subset of these servers are honest, this 
setup guarantees the same level of privacy as the central model. However, MPC 
techniques often introduce significant computational overhead, especially when 
distributing complex or expressive computations.

Another option is to avoid relying on external servers altogether. In this case, 
each client must publish their data in a locally differentially private manner, 
by adding noise directly to their data. While this removes the need for a trusted 
server, it results in lower utility due to the increased noise. This approach is 
known as the \emph{local model} of differential privacy~\cite{Kasiviswanathan}, an example 
of which is randomized response~\cite{warner_randomized_1965}.

Both the central and local models of differential privacy have strengths and 
drawbacks. The central model typically yields higher utility 
but requires strong trust in a single server or the costly use of MPC. The local 
model, by contrast, requires no trust in external parties but suffers from reduced 
utility. These tradeoffs have motivated significant research into finding a middle ground 
between these two models. One promising direction is to limit the expressiveness 
of the class of functions $\mathcal{F}$ that can be executed in a trusted manner. 
While this restriction 
allows for more efficient distributed implementation, it is only useful if 
the DP mechanism built around such functions can still achieve good accuracy.

A notable example of such restricted expressiveness is the class of functions known as 
\emph{shuffles}, where a central entity (or a distributed system) randomly permutes the 
messages from clients before they reach the analysing server. This is known as the 
\emph{shuffle model}~\cite{bittau_prochlo_2017,erlingsson_amplification_2019,Cheu_2019}. 
The shuffle model has been used to 
implement many differentially-private mechanisms with reasonable utility~\cite{erlingsson_amplification_2019,Cheu_2019}. Surprisingly, the shuffle model 
is more expressive than initially expected: with a trusted shuffler, any function 
can be securely computed in two rounds under the assumption of an 
honest-majority~\cite{BHNS20}.

We therefore ask the following:
\begin{question}
    What is the least expressive class of functions $\mathcal{F}$ that needs to 
    be securely implemented in order to achieve computationally efficient 
    distributed differential privacy with utility comparable to that of the 
    central model?
\end{question}

\subsection{Our Contribution}

In this work, we investigate the power of the \emph{linear-transformation model (LTM)} for differentially-private mechanisms.
In this model, the clients only have access to a trusted platform for performing arbitrary linear transformations of their collective inputs. Linear functions can be distributed extremely efficiently with secure computations. However, it is also known that the expressiveness of this class of functions is strictly limited. 

For a visualization of the architecture, we refer to Figure \ref{fig:threat}. We demonstrate the benefits of this linear transformation model (LTM) by showing that it can be used, in combinations with \emph{linear sketches} to construct private summaries of a data set with noise comparable to that of central differentially private mechanisms. 
As in the local model, each client adds some amount of noise to their data before passing it onto the linear sketch. We then release the linear transform of the data.
The benefit of having a trusted linear sketch is that the noise can be combined, thereby significantly reducing the amount of noise per client.
All of our mechanisms require only a single round of communication from clients to server.

We demonstrate the effectiveness of this model by exemplarily applying it to problems in numerical linear algebra. 
A particularly powerful and expressive class of linear sketches with numerous linear algebra applications \cite{Woodruff14} are Johnson-Lindenstrauss (JL) transforms.
In this paper, we focus on \emph{low rank approximation} and \emph{regularized linear regression}.
Our contribution is to show how to design mechanisms for these problems, with the most widely used JL-transforms as the underlying secure sketch. 
For these problems, we achieve utility comparable to that of the central model with local privacy guarantees, assuming the JL-transform is computed securely.

For both of the following problems, we have the option of achieving $\varepsilon$ differential privacy, via protocols based on the Laplace mechanism, as well as $(\varepsilon,\delta)$-differential privacy with slightly higher utility, via protocols based on the Gaussian mechanism.

\mypar{Low Rank Approximation.}
We are given a data matrix $\mdata\in\mR^{n\times d}$.  Our goal is to compute an orthogonal projection $\mx\in\mR^{d\times r}$ minimizing the error $OPT = \min_{\mx}\|\mdata-\mdata\mx\mx^T\|_F$, where $\|.\|_F$ denotes the Frobenius norm of a matrix. 
In the LTM, we give a $(\varepsilon,\delta)$ differentially private mechanism with 
multiplicative error $(1 + \alpha)$ and additive error $\tilde{O}(kd^{3/2}\alpha^{-3}\varepsilon^{-1}\log \delta^{-1})$. We also present a $\varepsilon$ differentially private mechanism  with multiplicative error $(1+ \alpha)$ and additive error $\tilde{O}(k^{3}d^{3/2}\alpha^{-6}\varepsilon^{-1})$.
State of the art central model mechanisms~\cite{DworkTT014,hardt2012} have no multiplicative error and additive error $O(k\sqrt{d} \cdot \varepsilon^{-1})$ , which is optimal in terms of $k$ and $d$. For the local model, we are only aware of the result by \cite{AroraBU18}, which has a multiplicative error of $O(\log n \log^3 d)$ and additive error $O(k^2n\cdot d\text{poly}(\varepsilon^{-1},\log \delta^{-1}))$. There are several related questions we will survey in more detail in Section~\ref{sec:relwork}.

\mypar{Ridge Regression.}
We are given a data matrix $\mdata \in \mR^{n \times d}$, 
a target vector $\predictor\in\mR^{n}$ and regularization parameter $\lambda>0$. Our goal is to find $\x\in \mR^{d}$ minimizing
$\min_\x \|\mdata\x-\predictor\|^2 + \lambda\cdot \|\x\|^2$, where $\|.\|$ denotes the Euclidean norm of a vector. 
Each row of the data matrix, as well as the corresponding entry of $\predictor$ resides with a client.
Assume that $\lambda \geq \text{poly}(\varepsilon^{-1},d,\log 1/\delta) $. 
In the LTM, there exists an $(\varepsilon,\delta)$ differentially private mechanism that can compute an $\hat{\x}$ with error $ (1+o(1)) \cdot (\|\mdata\x_{opt}-\predictor\|^2 + \lambda\cdot \|\x_{opt}\|^2) + \text{poly}(\varepsilon^{-1},d,\log 1/\delta)$.
Unlike the low rank approximation problem, we are not aware of other papers achieving comparable bounds in any model. Regression has been studied in the context of empirical risk minimization and some results also exist for comparing closeness of the computed predictor $\hat{\x}$ and that optimal $\x_{opt}$. While we view the low rank approximation mechanisms as our main result, we nevertheless believe this work can serve as a starting point for achieving bounds for ridge regression in terms of cost in any model including the central model.

\begin{table*} [tb]
  \caption{ SOTA utility bounds for various models for frequency estimation and low rank approximation. Our bounds for low rank approximation have a multiplicative error $O(1+\alpha_\mS)$, which we omitted for space.}
  \label{tab:f1}
	 \centering
		\begin{tabular}{|l|c|c|c|c|}
   \hline
			& Local & Shuffle &   LTM ($\textbf{this work}$) & Central \\
   \hline
			\midrule
			\multirow{2}{*}{Frequency Est} & $\tilde{O}(\sqrt{n})$\cite{BassilyNST20} & \multirow{2}{*}{$\tilde{\Theta}(\varepsilon^{-1})$ \cite{GhaziG0PV21}}  & \multirow{2}{*}{$\tilde{O}(\varepsilon^{-1})$}  & \multirow{2}{*}{$\tilde{\Theta}(\varepsilon^{-1})$}  \\
            &  $\Omega(\sqrt{n})$\cite{BassilyS15} & & & \\
   \hline
             \multirow{2}{*}{Low-Rank App} & $(\tilde{O}(\log n \log^3 d),O(k^2nd))$\cite{AroraBU18} & \multirow{2}{*}{N.A.} & \multirow{2}{*}{$\tilde{O}(\frac{kd^{3/2}}{\alpha_\mS^3} \varepsilon^{-1} )$}  & \multirow{2}{*}{$\tilde{\Theta}(k\sqrt{d}\varepsilon^{-1})$\cite{DworkTT014}} \\
             &  $\Omega(\sqrt{n})$\cite{BassilyS15}  & & & \\
             \hline
		\end{tabular}

	\end{table*}
Table~\ref{tab:f1} provides a comparison of the utility bounds for the problems in the local, central, shuffle and LTM models. 
To the best of our knowledge, there are no results on low rank approximation in the shuffle model.

As a warm-up to illustrate the potential inherent in the LTM, and to also highlight the comparison to the other models, we consider the basic problem of computing histograms and frequency of moments.

\section{Related Work}\label{sec:relwork}

There is substantial work on the shuffle model, which also aims at facilitating differential privacy in a distributed setting. We discuss the relationship between the LTM and the shuffle model in more detail in Section \ref{sec:LTM}.
For the problems we studied here, there is an abundance of prior work discussed as follows.

\paragraph{Low Rank Approximation:}
There are various ways in which one could formulate the low rank approximation problem.
The setting which is most important to us is the seminal paper \cite{DworkTT014}, who achieved a worst case additive error of the order $k\sqrt{d}$ for outputting an orthogonal projection matrix $V_k$ in the row space of $\mdata$. This bound is also optimal.
In the local model where each client holds a row of the data matrix $\mdata$, and we wish to output an orthogonal projection in the column space of $\mdata$, \cite{Upadhyay18} gave an algorithm with an additive error of the order $\sqrt{n}$, which matches the lower bound by \cite{BassilyS15} up to lower order terms. This is different from computing a projections in the row space of $\mdata$ considered by \cite{DworkTT014} and indeed our own results.
The only know local algorithm to do likewise in the local model is due to \cite{AroraBU18}, with a substantially larger multiplicative and additive error compared to our results.
To the best of our knowledge, no single round shuffle protocol improving over the local bounds is known.

Low rank approximation and its sister problem PCA, have seen substantial attention in settings more loosely related to our work. 
Much of the work on private low rank approximation~\cite{blum2005,DBLP:conf/nips/ChaudhuriSS12,DworkTT014,hardt2012,HR13,HP14,BDWY16,KT13} considers the data to be fixed, often using Gaussian noise. 
In addition, a strong spectral gap assumption, typically of the form that the first singular value is substantially larger than the second, is crucial to the analysis. 
We note that all algorithms operating with this assumption do not yield worst case bounds.

When there is no spectral gap, we compare to the central model exponential mechanism approach from \cite{DBLP:conf/nips/ChaudhuriSS12}, implemented by~\cite{KT13}. \cite{LKO22} also provide a solution to this problem; however, it is unfortunately computationally intractable.

\paragraph{Ridge Regression}
The previous work most related to ours is due to \cite{Sheffet19}. The author gave a mechanism that preserved the entire spectrum of a data matrix in a private manner and also showed that the returned regression vector $\hat{\x}$ is close to the optimal $\x$ assuming that the matrix is well conditioned. 

Most previous work \cite{KiferST12,CWZ21,BassilyST14,WangFS15,DBLP:conf/nips/WangG018,Wang2018NoninteractiveLP,JMLR:v24:21-0523,wang18,VTJ22,MKFI22,LKO22,zheng2017collect}  study private linear regression in the context of risk minimization, where the algorithm is given i.i.d. samples from some unknown distribution and aims at computing a solution with good out of sample performance. Even without privacy constraints risk minimization yields an additive error that depends on $n$. 
Another line of work studies varies loss functions, including regression in a Bayesian setting~\cite{MASN16,FGWC16,DNMR14}. 
We are not aware of any prior work on linear regression in the local DP model.

\paragraph{DP via Johnson-Lindenstrauss Transforms.}

The Johnson-Lindenstrauss lemma has been previously used in differential privacy. \cite{DBLP:conf/focs/BlockiBDS12} studies the DP guarantee a JL transform itself gives in the central model, in the context of cut queries and directional variance queries. \cite{DBLP:journals/jpc/KenthapadiKMM13} and \cite{DBLP:conf/pods/Stausholm21} use JL transforms together with Gaussian noise in the context of DP Euclidean distance approximation, to decrease the error. This was subsequently improved in \cite{Stausholm21}, using Laplace noise. The difference to our work is that the transform is used to reduce $d$ instead of $n$, by having every client apply it locally before adding noise. In \cite{DBLP:conf/soda/Nikolov23} the authors make use of JL transforms to achieve private query release, by applying it before releasing a bundle of queries. \cite{DBLP:conf/nips/GhaziK0MS23} studies pairwise statistics in the local model of differential privacy and uses JL transforms to reduce client-sided dimensionality.

\paragraph{Differentially Private Non-Oblivious Sketches. }

Non-Oblivious sketches, such as CountSketch, have also been used in combination with either distributed or central DP.  \cite{ZQRAW22,PT22} provide an analysis of differentially private CountSketch in the central model. \cite{MDC16} consider differentially private CountSketch and Count-Min Sketch in a distributed setting, but they reduce the dimensionality of each input, which can be done by users locally, rather than reducing the dependence on the number of users, as we do.

\section{Preliminaries}\label{sec:prelims}

\mypar{Notation. }Column vectors are written in bold lowercase letters $\predictor$ and matrices in bold uppercase letters $\mdata$. The transpose operator over vectors and matrices is $\predictor^T$ and $\mdata^T$. For any vector $\predictor$, we denote by $\lVert\predictor\rVert = \sqrt{\sum_i b_i^2}$ its $\ell_2$-norm. For any matrix $\mdata$, we denote by $\lVert\mdata\rVert_F =\sqrt{\sum_i\sum_j \mdata_{i,j}^2}$ its Frobenius norm. We denote the inner product between vectors $\langle\data,\predictor\rangle=\data^T\predictor$. Two vectors $\data,\predictor$ are orthogonal if their inner product is 0. An orthogonal matrix $\mdata$ is a real square matrix where the columns have unit Euclidean norm and are pairwise orthogonal.

\mypar{Subspace Preserving Sketches. }We consider subspace-preserving sketches in this work, which are sparse Johnson-Lindenstrauss transforms. Of particular interest for us are algorithms that compute a subspace approximation of $A$ without prior knowledge of $A$. Such algorithms are known as oblivious subspace embeddings. We will use a family of such embeddings known as \emph{oblivious sparse norm-approximating projections} originally due to \cite{nelson2013osnap}, but since improved upon in several subsequent works.

\begin{definition}[OSNAP \cite{nelson2013osnap}.]
\label{def:osnap}
    Let $\mathcal{D}_{m,n,\decom}^{\algofont{sketch}}$ be a distribution over random matrices $\{-1,0,1\}^{m\times n}$, 
    where $\decom$ entries of each column, chosen uniformly and independently for each column, are set to $\pm 1$ with equal probability, and all other entries are set to 0. If $\decom=1$, we write $\mathcal{D}_{m,n}^{\algofont{sketch}}$. We say that $S\sim \mathcal{D}_{m,n,\decom}^{\algofont{sketch}}$ is an $(\alpha,\beta,\decom,m)$ OSNAP for a $n\times d$ orthogonal matrix $\mathbf{W}$ if with probability $1-\beta$, for all $a,b\in \mathbb{R}^d$ and some precision parameter $\alpha$
    \[\left\vert a^T\mathbf{W}^T\mS^T\mS\mathbf{W}b - a^T\mathbf{W}^T\mathbf{W}b\right\vert \leq \alpha\cdot \|\mathbf{W}a\|\|\mathbf{W}b\|.\]
\end{definition}

We give bounds on available choices of $(\alpha,\beta,\decom,m)$ for subspaces of rank $k$ in Table~\ref{tab:ose}.

\begin{table*}
    \caption{Trade-off between the target dimension $m$ and sparsity $s$ of sparse OSEs that are generated such that they have $s$ non-zeros entries per column. Here $\alpha$ denotes the accuracy of the OSE and $\beta$ denotes its fail probability. Variable $k$ is the parameter to $k$-rank approximation, which can be replaced by dimension $d$ for linear regression. The paper by \cite{Achlioptas03} was the first to propose the construction, while its application to preserving subspaces was given by \cite{ClarksonW09}.}
    \label{tab:ose}
    \centering
        \begin{tabular}{|c|c|c|c|c|c|}
            \hline
            \textbf{Paper} & \textbf{Target Dimension} $m$ & \textbf{Sparsity} $s$ & \textbf{With Probability} \\
            \hline
            \cite{Achlioptas03} & $O(\frac{k+\log 1/\beta}{\alpha^2})$ & $m$ & $1 - \beta$ \\
            \hline
            \cite{nelson2013osnap} & $O(\frac{k^2}{\alpha^2\beta})$ & 1 & $1 - \beta$ \\
            \hline
            \cite{Cohen16} & $O(\frac{k\log(k/\beta)}{\alpha^2})$ & $O(\frac{\log(k/\beta)}{\alpha})$ & $1 - \beta$  \\
            \hline
            \cite{schwiegelshohn23} & $O(\frac{k}{\alpha^2})$ & $O(\frac{1}{\alpha} \cdot (\frac{k}{\log(1/\alpha)} + k^{2/3}\log^{1/3}(k)))$ & $1 - 2^{-k^{2/3}}$ \\
            \hline
        \end{tabular}
\end{table*}

Oblivious subspace embeddings (OSEs)~\cite{DBLP:conf/focs/Sarlos06}, allow for faster approximation algorithms for problems in linear algebra. Achlioptas \cite{Achlioptas03} provided the first Johnson-Lindenstrauss transform with some amount of sparsity. The first embedding with asymptotically smaller number of non-zeros than dense Johnson-Lindenstrauss transforms was probably due to \cite{DBLP:conf/soda/KaneN12}, who also applied them in the context of subspace embeddings. In a remarkable result, the super-sparse version with only $1$ non-zero entries per column, was first analyzed in~\cite{ClarksonW13} and improved independently in~\cite{MengM13,nelson2013osnap}, by increasing the target dimension $m$. In~\cite{nelson2013osnap}, the sparse Johnson-Lindenstrauss transform~\cite{DBLP:conf/soda/KaneN12} was studied with a wider range of parameters. Subsequent works~\cite{Cohen16, schwiegelshohn23} further analyze and improve the relationship between $m$ and $s$, emphasizing different parameters. Table \ref{tab:ose} provides the exact interplay between $m$ and $s$ in those works. \cite{DBLP:conf/nips/LiangBKW14} proposes an approach to boost the success probability $\beta$ of an OSE, which gives an alternative to increasing the target dimension.

When dealing with rank $k$ approximation, we will condition on 
\begin{align*}
    (1-\alpha_S) \| \mdata - \mdata\mx\mx^T \|_F^2 & \leq \|\sketch(\mdata - \mdata\mx\mx^T)\|_F^2 \\
    & \leq (1+\alpha_S) \| \mdata-\mdata\mx\mx^T \|_F^2
\end{align*}
% $$$$
for all rank $k$ orthogonal matrices $\mx$.
When dealing with regression, we will condition on
\begin{align*}
    (1-\alpha_S) \| \mdata \x'-\predictor \|^2 & \leq \|\sketch(\mdata \x-\predictor)\|^2 \\
    & \leq (1+\alpha_S) \| \mdata \x'-\predictor \|^2
\end{align*}
for all $\x \in \mathbb{R}^d$. 
For the parameters given here, this is true with the probability given in Table \ref{tab:ose} (assuming $k=d$ in the case of regression).
For the case of regression in particular, it is sometimes beneficial to select $\alpha_S$ as large as possible. A sufficiently largest value of $\alpha_S$ such that $\sketch$ still provide a subspace embedding guarantee is $1/3$. Throughout this paper, we will sometimes bound $\alpha_S$ by $1$.

We further will give our utility proofs using the sparsity/target dimension bounds from \cite{Cohen16}. Other tradeoffs are possible, but they have worse bounds for most ranges of parameters. The results in \cite{MengM13,nelson2013osnap,schwiegelshohn23} give different trade-offs depending on which set of parameters are considered the most dominant. Notably, the $s=1$ sketches of \cite{MengM13,nelson2013osnap} result in an additive error of $\Tilde{O}\left(\frac{k^4d^3}{\alpha_S^4 \beta^2} \varepsilon^{-2}\log \frac{1}{\delta}\right)$ for low rank approximation and a term $\Tilde{O}\left(\frac{d^7\varepsilon^{-2}\log \frac{1}{\delta}}{\alpha_S^5 \beta^2 \lambda} + \frac{d^{14}\varepsilon^{-4}\log^2 \frac{1}{\delta}}{\alpha_S^6 \beta^4 \lambda^2}\right)$
in both the multiplicative and additive error. The utility guarantees using the bounds from \cite{schwiegelshohn23} change only by logarithmic factors compared to those in Theorems \ref{thrm:utilitykrank} and \ref{thrm:linreg}.

\mypar{Differential Privacy. } Differential privacy~\cite{DMNS06} offers privacy guarantees to individuals contributing their data to some randomized algorithm. We say that two datasets are neighboring if one can be obtained from the other by the the replacement of a single individual with another individual.

\begin{definition}[Differential Privacy in the central model \cite{DMNS06}]\label{def:DP}
    Let $\epsilon \geq 0$ and $\delta \in [0,1]$. A randomized mechanism $\mathcal{M}:\mathcal{X} \to \mathcal{Y}$ is $(\epsilon,\delta)$-differentially private, if for all neighboring data sets $x, x' \in \mathcal{X}$ and all outputs $S \subseteq \mathcal{Y}$ it holds that
    \begin{equation*}
        \Pr[\mathcal{M}(x) \in S] \leq e^\epsilon \Pr[ \mathcal{M}(x') \in S ] + \delta,
    \end{equation*}
    where the probabilities are  over the randomness of $\mathcal{M}$. If $\delta=0$, then we simply say that the mechanism is $\varepsilon$-differentially private.

\end{definition}

We first define the Gaussian and Laplace mechanisms, which we use later. 

\begin{lemma}[The Gaussian Mechanism \cite{dwork2014}]\label{lemma:gaussian_mech}
    Let $f:\mathcal{X} \to \R^k$ be a function and let $\epsilon \geq 0$ and $\delta \in [0,1]$. The Gaussian mechanism adds to each of the $k$ components of the output, noise sampled from $N(0,\sigma^2)$ with
    \begin{equation*}
        \sigma^2 \geq \frac{2(\Delta_2f)^2\ln(1.25/\delta)}{\epsilon^2},
    \end{equation*}
    where $\Delta_2f = \max_{x \sim x'} \lVert f(x) - f(x') \rVert_2$ denotes the $\ell_2$ sensitivity of function $f$. 
    The Gaussian mechanism is $(\epsilon,\delta)$ differentially private.
\end{lemma}

\begin{lemma}[The Laplace Mechanism \cite{DMNS06}]\label{lemma:lap_mech}
    Let $f:\mathcal{X} \to \R^k$ be a function and let $\epsilon \geq 0$. The Laplace mechanism adds to each of the $k$ components of the output, noise sampled from $Lap(\Delta_1f/\varepsilon)$, where $\Delta_1f = \max_{x \sim x'} \lvert f(x) - f(x') \rvert$ denotes the $\ell_1$ sensitivity of function $f$. 
    The Laplace mechanism is $\epsilon$ differentially private.
\end{lemma}

We will use sequential composition of differentially private mechanisms.
% Theorem 3.16 in [Dwork]
\begin{lemma}[Sequential Composition \cite{DMNS06}]\label{lemma:composition}
    Let $\mathcal{M}_i:\mathcal{X} \to \mathcal{Y}_i$ be an $(\epsilon_i,\delta_i)$-differentially private mechanism for $i \in [k]$. Then mechanism $\mathcal{M}:\mathcal{X} \to \prod_{i=1}^k \mathcal{Y}_i$ defined as $\mathcal{M}(x) = (\mathcal{M}_1(x),\dots , \mathcal{M}_k(x))$, is $(\sum_{i=1}^k \epsilon_i, \sum_{i=1}^k \delta_i)$-differentially private.
\end{lemma}

A probability distribution $F$ is \emph{infinitely divisible} if for any positive integer $n$, there is a random variable $S_n$ which can be written as
$S_n = \sum_{i\in[n]} X_{i}$
such that each of the $X_i$ are random variables, independent and identically distributed, and $S_n$ has the probability distribution $F$.

\section{Privacy Guarantees in the LTM}
\label{sec:LTM}

We consider a model in which the trusted component can perform any public linear transformation of the inputs. 
The model includes three algorithms: (1) $R:\mathcal{X}\rightarrow\mathcal{Y}$ is a randomized encoder that takes a single user's data and outputs a randomized message, (2) $T:\mathcal{Y'}\rightarrow\mathcal{Y^*}$ is the idealized trusted component that performs a public linear transformation of its inputs, and (3) $A:\mathcal{Y^*}\rightarrow\mathcal{Z}$ is an analysis function that takes the results messages and estimates some function from these messages.

Note that standard definitions of differential privacy in the shuffle model only require $T(R(x_1),...,R(x_n))$ to be differentially private. This implicitly assumes that all clients are honest and do not collude with the adversary, in particular they are assumed not to leak the output of their randomizers publicly. This implies that existing definitions of differential privacy in the shuffle model could be satisfied even by (artificial) mechanisms in which a single client adds the whole noise while the others do not randomize their messages at all. This is clearly a very weak privacy guarantee: in a setting in which a large number of clients participate in a differentially private data analysis, it is unrealistic to assume that the adversary does not control even a single client.  Luckily, to the best of our knowledge, no proposed mechanisms in the literature suffer from these vulnerabilities, still this counterexample shows that the existing definition is not robust enough, and we therefore formalize a notion of differential privacy where we explicitly tolerate that a bounded number of clients might collude with the adversary. Similar observations were also made by~\cite{talwar2023samplable}.

\begin{definition}[Trusted Computation Model for Differential Privacy]\label{def:trustDP}
        A tuple of algorithms $P=(R,T,A)$ is $(\epsilon,\delta)$-differentially private given corrupt clients $\corruptc$ if the output $\Pi_R(x_1,...,x_n)=T(R(x_1),...,R(x_n))$, as well as corrupted parties' randomizer output $R(x_i)$ for all $x_i\in \corruptc$ satisfy $(\epsilon,\delta)$-differential privacy. 
\end{definition}

\mypar{Multi-Central Model of Differential Privacy.} 

\cite{steinke} introduce the multi-central model, which exactly defines this split-trust model that instantiates the trusted computation model above. They allow a semi-honest adversary that honestly follows the protocol to corrupt up to all but one servers and all but one client. We also operate in the same semi-honest setting, but define parameters $t$ and $t'$ for the number of tolerated server and client corruptions.

 The view $\algofont{View}^{R,\Pi,A}_{\corruptc,\corrupts}(x)$ of the adversary consists of all the information available to the corrupted clients  $\corruptc$ and servers $\corrupts$ plus the final output of the honest, uncorrupted servers. This view excludes only internal information of the trustworthy servers and clients. The resulting protocol is $(\epsilon,\delta)$-differentially private in the multi-central model if the adversary's view is $(\epsilon,\delta)$-indistinguishable from the output on a neighboring dataset. 

 Figure~\ref{fig:threat} further illustrates what is contained in the adversary's view. 
 If one server and one client are corrupted, the adversary's view contains all incoming and outgoing messages from that server and that client marked in red in the figure.

 \begin{figure}
     \centering
     \includegraphics[width=\linewidth]{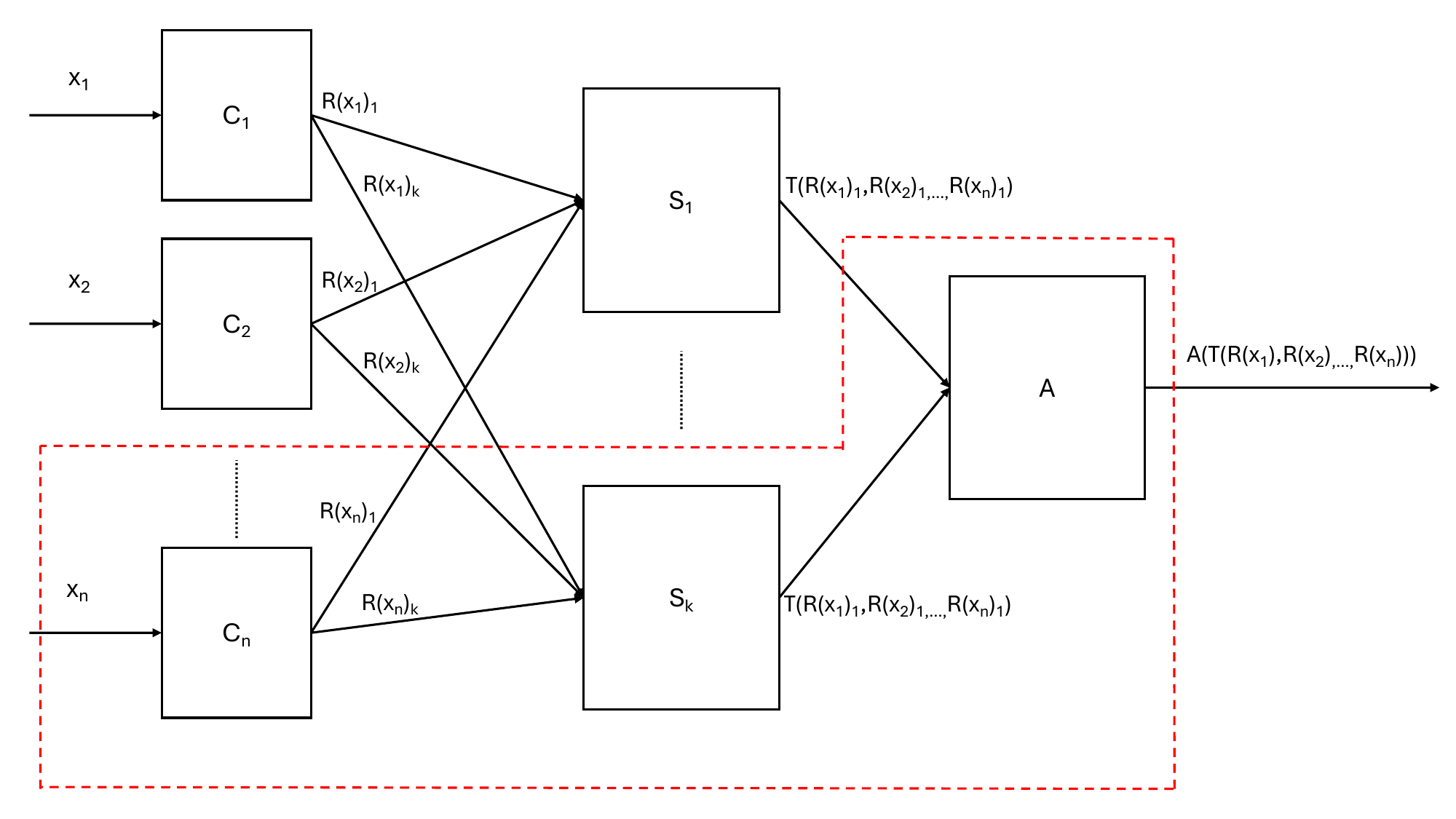}
     \caption{Adversary's View}
     \label{fig:threat}
 \end{figure}

\begin{definition}[Instantiation of Trusted Computation Model for Differential Privacy with MPC]\label{def:MPCDP}
        Let $\Pi$ be a $k$-party MPC protocol that computes $f:\mathbb{R}^n\rightarrow \mathbb{R}$ with perfect security, tolerating $t$ corruptions. 
        A tuple of algorithms $P=(R,\Pi,A)$ is $(\epsilon,\delta)$-differentially private if for all coalitions $\corrupts$ of up to $t<k$ corrupt servers and for all possible $S$, all coalitions $\corruptc$ of $t'<n$ corrupt clients and all neighboring datasets $x$ and $x'$:
        \begin{align*}
                     \Pr[\algofont{View}&^{R,\Pi,A}_{\corrupts,\corruptc}  (x)\in S] \\
        & \leq e^\epsilon\cdot \Pr[\algofont{View}^{R,\Pi,A}_{\corrupts,\corruptc}(x')\in S]+\delta
        \end{align*}

        % given corrupt parties $C$ if the output $\Pi_R(x_1,...,x_n)=S(R(x_1),...,R(x_n))$, as well as corrupted parties' randomizer output $R(x_i)$ for all $x_i\in C$ satisfy $(\epsilon,\delta)$ differential privacy. 
        where the probability is over all the randomness in the algorithms $(R,\Pi,A)$. 

\end{definition}

Algorithm~\ref{alg:ltm} describes how clients and servers jointly compute the linear transformation, allowing the transformed data to be published and used for analysis. 
\begin{algorithm}[tb]
	\caption{\expandafter Linear Transformation Model}
        \begin{flushleft}
        \begin{algorithmic}[1]
		\STATE \textbf{Input:} Individual input vectors $\mathbf{x}_{1},\dots,\mathbf{x}_{n}$ representing values in $\mathbb{R}^d$ %\linebreak
        \STATE Each client $i\in[n]$ locally computes $R(\mathbf{x}_i)$ on their input $\mathbf{x}_i$.
        \STATE Each client secret shares the resulting value using a linear secret sharing scheme, and sends one share to each server.
        \STATE The servers jointly compute function $T$ on the matrix, which results from concatenating all $n$ vector secret shares.
        \STATE The resulting matrix product is published and then taken as input to any analysis function $A$.
        \STATE  \textbf{Output:} Output $A(T(R(\mathbf{x}_1),\dots,R(\mathbf{x}_n))$
        \end{algorithmic}\label{alg:ltm}
        \end{flushleft}
\end{algorithm}

With this definition, we can give our first result showing that differential privacy is retained given a bounded number of corruptions.

\begin{lemma}\label{lem:mpcpriv}
    If  $P=(R,T,A)$ is $(\epsilon,\delta)$-differentially private with $t'$ corrupt clients% $\corruptc$
    , and if $\Pi$ is a perfectly secure $k$-party MPC protocol that computes $T$ correctly while tolerating $t$ corrupt servers% $\corrupts$
    , then $P=(R,\Pi,A)$ is $(\epsilon,\delta)$-differentially private.
\end{lemma}

\begin{proof}[Proof Sketch.] This proof can be seen as analogous to the security proof of Theorem 4 in~\cite{talwar2023samplable}. Recall that each server learns nothing more from a perfectly secure multi-party computation protocol other than what is implied by their own input and the output of the function being evaluated. Then there is a simulator $\algofont{Sim}^{R,\Pi,A}_{\corrupts,\corrupts=c}(x)$ whose output is $(\epsilon,\delta)$ indistinguishable from the adversary's view $\algofont{View}^{R,\Pi,A}_{\corrupts,\corruptc}(x)$. Note that the secret shares that are part of the view of the corrupted servers in the protocol $\Pi$ are independent of the input values of the corrupted clients, and all can be easily simulated by random sampling.

A bit more formally, let $x$ be an input dataset, where $x_i$ is the input value contributed by client $i$ and $\Tilde{x}_i$ is the vector of input shares for server $i$ that serves as input to $\Pi$. There is a mapping from $x\in\mathcal{X}^n$ to $\Tilde{x}\in\Tilde{\mathcal{X}}^k$. Denote $X=\{x_i\}_{i\in \corruptc}$ and $X'=\{x'_i\}_{i\in \corruptc}$.Then for any two datasets $x$, $x'$ that differ in one entry $x_i$ with $i \notin \corruptc$:

\begin{align*}
    & \Pr[  \algofont{View}^{R,\Pi,A}_{\corrupts,\corruptc}(x)\in S] \\
    & \leq \Pr[ \algofont{Sim}^{R,\Pi,A}_{\corrupts,\corruptc}(f(R(x_1),\dots,R(x_n),X)\in S] \\
    & \leq e^\epsilon \Pr[  \algofont{Sim}^{R,\Pi,A}_{\corrupts,\corruptc}(f(R(x'_1),\dots,R(x'_n),X')\in S] + \delta \\
    & = e^\epsilon  \Pr[  \algofont{View}^{R,\Pi,A}_{\corrupts,\corruptc}(x')\in S] + \delta 
\end{align*}
\end{proof}

\subsection{Frequency Moments}
\label{sec:frequency}
To illustrate the possibilities inherent to the LTM, we consider a simple application by way of estimating frequency moments. For simplicity, assume in this section that all clients are trustworthy, i.e. $t'=0$, but the arguments can be straightforwardly extended to deal with arbitrary values of $t'$.
Here, each client is given a number $[-\Delta,\Delta]$ and our goal is to estimate $F_k = \sum_{i=1}^n |x_i|^k$.
In the local model, even if all entries are either $0$ or $1$, it is not possible to estimate $F_1$ without incurring an additive error of the order $\sqrt{n}$, that is the estimated value $\Tilde{F_1} = F_1 \pm O(\sqrt{n})$ for any privacy preserving local mechanism \cite{BassilyS15}. In contrast, simply computing $F_k$ and adding an appropriate amount of noise %\footnote{Using the Gaussian mechanism, $\sigma^2 = \rho\cdot \frac{\Delta^k}{\varepsilon^2}\log \delta^{-1}$ for some absolute constant $\rho$ is sufficient, though other mechanism can yield even better bounds.}
yields an additive error that only depends on $\Delta$, $\varepsilon$ and $\delta$.

We now consider the LTM.
Let $\mathbf{x}^k$ denote the vector of client entries with $\mathbf{x}^k_i = |x_i^k|$. We observe that $F_k =f_k(\mathbf{x})= 1^T \mathbf{x}^k$. Suppose that every client $i$ samples $g_i$ from a distribution, and we denote $\tilde{\mathbf{x}}$ the vector of noisy client entries with $\tilde{\mathbf{x}}^k_i = |x_i^k|+g_i$. 
To guarantee differential privacy, we can use a differentially private additive noise mechanism $1^T (\tilde{\mathbf{x}}^k ) = f_k(\mathbf{x})+G$, where the noise distribution is infinitely divisible and $G=\sum_{i\in[n]}g_i$.

Many different infinitely divisible distributions can be used at this stage. One example is the geometric distribution, since the sum of negative binomial random variables follows a geometric distribution, which can be used in differentially private mechanisms \cite{Ghazi0MP20}. The Laplace distribution, which is commonly used for $\epsilon$-DP mechanisms, is also infinitely divisible, since a zero-centered Laplace distribution can be formulated as a sum of differences between Gamma distributions. Other possible options are stable distributions, a subset of infinitely divisible distributions for which each $g_i$ and $G$ follow the same distribution, such as Gaussian or Cauchy. 

Since Gaussians are popular in differentially private mechanisms, we will describe the parameter choices for this distribution. Let $g_i$ be Gaussian distributed with zero mean and variance $\sigma^2 \geq \frac{2\Delta^k \ln(1.25/\delta)}{n\cdot \varepsilon^2}$. Then we have $\tilde{F_k} = 1^T (\mathbf{x}^k + g) = F_k + \sum_{i=1}^n g_i$. We now observe that $\sum_{i=1}^n g_i$ is Gaussian distributed with zero mean and variance $n\cdot \sigma^2 \geq \frac{2\Delta^k \ln(1.25/\delta)}{\varepsilon^2}$. Thus, assuming a trusted computation of $1^T g$, $\tilde{F_k}$ is a differentially private estimate of $F_k$ with error $\frac{2\Delta^k \ln(1.25/\delta)}{\varepsilon^2}$.

We also describe the analogous parameter choices for the Laplace mechanism. Let $g_i$ be sample from the distribution defined as the difference between two Gamma distributions $\Gamma(\frac{1}{n},\frac{2\Delta^k}{\epsilon})$. Again, we have $\tilde{F_k} = 1^T (\mathbf{x}^k + g) = F_k + \sum_{i=1}^n g_i$. We now observe that $\sum_{i=1}^n g_i$ is Laplace distributed, centered at 0 and with scale $\frac{2\Delta^k}{\epsilon}$. Thus, assuming a trusted computation of $1^T g$, $\tilde{F_k}$ is a differentially private estimate of $F_k$ with error $\frac{2\Delta^k }{\varepsilon}$.

\subsection{Privacy Preserving Mechanism for Oblivious Johnson Lindenstrauss Transforms}

We instantiate the trusted computation model for differential privacy for specific choices for the functions in tuple $(R,T,A)$. We will show how to guarantee differential privacy for the tuples $(R_D^{\text{dense}},T_\mS^{\text{dense}},A)$ and $(R_D,T_\mS,A)$, which we will refer to as the dense sketching mechanism and the sparse sketching mechanism, respectively.

\mypar{Randomizer function $R_D:\mR^d\rightarrow \mR^{d\decom}$.} 
The randomizer takes as input a $d$-dimensional row vector $\mathbf{x}$, copies it $s$ times and concatenates the result to form $\mathbf{x}^s$. Then noise is sampled independently from distribution $D$ and added to each of the $ds$ entries in the resulting vector.
% We define two choices for: %$R_{D}$ and $R^{\text{dense}}_D$. 

\begin{equation}\label{eq:Rdense}
    R^{\text{dense}}_D(\mathbf{x}) = 
      \mathbf{x}^\decom+\noise, \noise\sim D^{d\decom} 
\end{equation}
We define a second alternative to this randomizer function, which instead outputs $0^{d\decom}$ if $n<8m\ln(dm/\delta)+t'$:
\begin{equation}\label{eq:R}
    R_{D}(\mathbf{x}) = \begin{cases} 
      0^{d\decom} & n<8m\ln(dm/\delta)+t' \\
      R^{\text{dense}}_D(\mathbf{x}) & else 
   \end{cases}
\end{equation}

\mypar{Transformation function $T_\mS:\mR^{n\times d\decom}\rightarrow\mR^{m\times d}$.} 
This function takes as a parameter a matrix $\mS\in\{-1,0,1\}^{m\times n}$ and takes as input a matrix $\mA\in\mR^{n\times d\decom}$, which corresponds to the vertical concatenation or stacking of $n$ outputs from a randomizer function $R$. We parse $\mA$ as the horizontal concatenation of matrices $\mA_i\in\mR^{n\times d}$ for $i\in[\decom]$. Let $\decom$ be the largest number of non-zero entries in a column of $\sketch$. Then $\sketch$ can be decomposed such that $\sketch = \frac{1}{\sqrt{\decom}}\sum_i^\decom \sketch_i$, where $\sketch_i\in\{-1,0,1\}^{m\times n}$ is a matrix with only one non-zero entry per column, and the assignment of non-zero entries to indeces $i$ is done uniformly at random. $T_\mS$ outputs the sum of matrix products $\frac{1}{\sqrt{\decom}}\sum_{i\in[\decom]}\mS_i\mA_i$.
\begin{equation}\label{eq:Ts}
    T_\mS (\mA) = \frac{1}{\sqrt{\decom}}\sum_{i\in[\decom]}\mS_i\mA_i
\end{equation}
We also define $T_\mS^{\text{dense}}:\mR^{n\times dm}\rightarrow\mR^{m\times d}$, where $\mS$ is a dense matrix rather than a sparse one, so we set $\decom=m$. Then $\sketch$ can be decomposed such that $\sketch = \frac{1}{\sqrt{m}}\sum_i^m \sketch_i$, where $\sketch_i\in\{-1,0,1\}^{m\times n}$ is a matrix with only one non-zero entry per column and exactly $n/m$ non-zero entries per row. 

\mypar{Analysis function $A$.} $A$ takes a matrix in $\mR^{m\times d}$, as well as any public values, and outputs some data analysis. Examples include low rank approximation and linear regression.

We now evaluate the noise distribution necessary to guarantee differential privacy in this model.
We find that for sufficiently large $n$, we can apply the infinite divisibility of distributions to divide the noise into a number of pieces corresponding the number $(n-t')/2m$ of honest clients expected to contribute to any entry of a column in the resulting noisy matrix product. The concentration bound, formulated in Lemma~\ref{lem:nonzero} and used to bound this number of honest clients, contributes to the necessary value of $\delta$. %If $n$ is small, we instead allow each client to add enough noise locally to make the resulting linear transformation differentially private. 
Composition theorems can be applied on the columns of the resulting matrix product to yield the final differential privacy guarantee.
We will later see that adding suitable noise to data results in only a small error in the resulting error for regularized linear regression and low rank approximation.

We are now ready to state our privacy guarantees in the LTM.

\begin{theorem}\label{thm:privacynew_general}
Let $\epsilon \geq 0$, $\delta \in (0,1)$, $m\in[n]$, $t'<n$, $\decom\in [m]$. % and 
            Let $\mS\sim \mathcal{D}_{m,n,\decom}^{\algofont{sketch}}$ with $\decom$ non-zero entries per column.
            Also let noise $g_i$ be sampled from a symmetric distribution $D$ such that $\sum_{i\in[(n-\decom-t')/(2m)]}g_i=G$. $G$ follows a distribution $D'$ that defines an additive noise mechanism that satisfies $(\epsilon/(\decom d), \delta/(\decom d) - m\exp( - (n-\decom -t') / 8m )) )$-DP for function $f(\mdata)=(\sketch\mdata)_j+\mathbf{n}$ for $\mathbf{n}\sim D'^m$ on input $\mdata\in\mathbb{R}^{n\times d}$, where the subscript refers to one column of $\sketch\mdata$. 
            Then as long as input values are bounded above by $\eta$, the sparse sketching mechanism $(R_D,T_\mS,A)$, as defined by equations~\ref{eq:R} and~\ref{eq:Ts} and for any $A$, is $(\epsilon,\delta)$-differentially private in the trusted computation model for differential privacy with $t'$ corrupt clients.
\end{theorem}

\begin{proof}
    We begin by proving a helper lemma below.
    \begin{lemma}\label{lem:nonzero}
    Let $\gamma \in (0,1)$, and let $\sketch \sim \mathcal{D}_{m,n}^{\algofont{sketch}}$. Then the probability that $\sketch$ has at least $(1-\gamma)\frac{n}{m}$ and at most $(1+\gamma)\frac{n}{m}$ non-zero entries in every row is bounded by; 
    
    \begin{equation*}
        \Pr\Bigl(\exists i \in [m]:\left\vert X_i - \frac{n}{m}\right\vert > \gamma\cdot \frac{n}{m}\Bigr) < 2m \exp\Bigl(-\frac{\gamma^2 n}{2m}\Bigr),
    \end{equation*}
    where $X_i$ denotes the number of non-zero entries in row $i$.
    \end{lemma}
    \begin{proof}
    We only give a proof of the lower bound, as a proof of the upper bound is completely analogous.
    Let $X_{ij}$ denote the indicator random variable that is $1$ if $\sketch_{ij}$ is non-zero and $0$ otherwise. Note that each entry of $\sketch$ can be thought of as an independent Bernoulli random variable $X_{ij}$, and the number of non-zero entries in a row is the sum of $n$ of these random variables. Fixing a row $k \in [m]$, we can use a Chernoff bound for sums of Bernoulli random variables to get a concentration bound around the expected number of non-zero entries $X_k = \sum_{j=1}^{n} X_{kj}$ in row $k$. In order to do this, we first need the expected value of $X_k$;
    \begin{equation*}
        \E[X_k] = \sum_{j=1}^{n} \Pr(X_{kj} = 1) = \frac{n}{m}.
    \end{equation*}
    Thus a Chernoff bound gives us that
    \begin{equation*}
        \Pr\Bigl(X_k < (1-\gamma) \frac{n}{m}\Bigr) < \exp\Bigl(-\frac{\gamma^2 n}{2m}\Bigr).
    \end{equation*}
    Applying a union bound over all rows of $\sketch$ then gives us exactly what we were to prove;
    \begin{align*}
        \Pr\Bigl( & \exists i \in [m]: X_i < (1-\gamma)\frac{n}{m}\Bigr) \\
        &\leq \sum_{k = 1}^{m} \Pr\Bigl(X_k < (1-\gamma) \frac{n}{m}\Bigr) \\
        & < m \exp\Bigl(-\frac{\gamma^2 n}{2m}\Bigr).
    \end{align*}
    \end{proof}
    
    To prove Theorem~\ref{thm:privacynew_general}, we first prove a simpler version of the theorem below.
    
    \begin{lemma}\label{thm:privacy_inf_div}
        Let $\epsilon \geq 0$, $\delta \in (0,1)$, $m\in[n]$, $t'<n$. % and 
    
        Let $\mS\sim \mathcal{D}_{m,n}^{\algofont{sketch}}$ with one non-zero entry per column.
        Also let $g_i\sim D$ be samples from a distribution such that $\sum_{i\in[(n-t')/(2m)]}g_i=G$. 
        $G$ follows a distribution $D'$ that defines an additive noise mechanism that satisfies $(\epsilon/d, \delta/d - m\exp( - (n-t') / 8m )) )$-DP for function $f(\mdata)=(\sketch\mdata)_j+\mathbf{n}$ with $\mathbf{n}\sim D'^m$ on input $\mdata\in\mathbb{R}^{n\times d}$, where the subscript refers to one column of $\mS\mdata$. %, assuming at least $\frac{n-t'}{2m}$ nonzero entries per row of $\mS$ corresponding to uncorrupted clients.
        Then as long as input values are bounded above by $\eta$, tuple of algorithms $(R_D,T_\mS,A)$, as defined by equations~\ref{eq:R} and~\ref{eq:Ts} and for any $A$, is $(\epsilon,\delta)$-differentially private in the trusted computation model for differential privacy with $t'$ corrupt clients.% $C'$.
    \end{lemma}
    
    \begin{proof}[Proof of Lemma~\ref{thm:privacy_inf_div}]
        Let $\epsilon$, $\delta$, $R_D$ and $T_\mS$ be given as described in the theorem. 
        If $\delta/d \leq m\exp( - (n-t') / 8m )$, then the $R_D$ outputs $0^d$, which is entirely independent of the input and is therefore $(\epsilon, \delta)$ differentially private. 
        
        Otherwise, $\mathcal{M}(\mdata) = T_\mS(R_D(\data_1), \dots , R_D(\data_n))=\mS(\mdata+\mnoise)$ where $\data_i$ denotes row $i$ of $\mdata\in\mR^{n\times d}$ and $\mnoise\leftarrow D^{n\times d}$, and we now prove that this is $(\epsilon,\delta)$ differentially private.
    
        Note that $\sketch_i(\mdata+\mnoise)$ is a d-dimensional vector, when $\sketch_i$ refers to the $i$-th row of $\sketch$.
        If $(\sketch_i(\mdata+\mnoise))_j$ is $(\epsilon/d,\delta/d)$-DP, then sequential composition (Lemma \ref{lemma:composition}) gives that $\mathcal{M}$ is $(\epsilon,\delta)$ differentially private.
    
        Now consider $d=1$ and let mechanism $\mathcal{M}_1:\mathbb{R}^n \to \mathbb{R}^m$ be defined as $\mathcal{M}_1(\mathbf{a}) = \mathbf{S} (\mathbf{a} + \mathbf{g})$ for $\mathbf{a}\in\mathbb{R}^n$, with $\mathbf{g}\leftarrow D^n$. 
        
        $\mathbf{S}$ is sampled such that it has only one non-zero entry per column; thus, the columns of $\mS$ are orthogonal. This means that the entries of $(\mathbf{S}(\mathbf{a} + \mathbf{g}))_j$ and $(\mathbf{S}(\mathbf{a} + \mathbf{g}))_{j'}$ for $j\neq j'$, have disjoint support. Therefore, the privacy guarantee can be analyzed independently for each entry of $\mathbf{S} (\mathbf{a} + \mathbf{g})$.

        Denote by $E$ the event that the support of $\sketch(a+g)_j$ has at least $\frac{n-t'}{2m}$ uncorrupted clients, whose set of indices we denote $N$. 
         We can formulate $\mathbf{S} (\mathbf{a} + \mathbf{g})_j=(\mathbf{S} \mathbf{a})_j + (\mathbf{S}\mathbf{g})_j = (\mathbf{S} \mathbf{a})_j + \sum_{k\in[n]}\mathbf{S}_{j,k}\mathbf{g}_k = (\mathbf{S} \mathbf{a})_j + \sum_{k\in N}\mathbf{S}_{j,k}\mathbf{g}_k + \sum_{k\in [n]\setminus N}\mathbf{S}_{j,k}\mathbf{g}_k$. Also note that $\sum_{k\in N}\mathbf{S}_{j,k}\mathbf{g}_k$ contains at least $\frac{n-t'}{2m}$ terms in event $E$, the sum of which suffice for guaranteeing DP.
        
        Due to the fact that $D$ is symmetric, $-1\cdot\mathbf{g}_k$ follows the same distribution as $\mathbf{g}_k$, and $\sum_{k\in N}\mathbf{S}_{j,k}\mathbf{g}_k$ follows the same distribution as $\sum_{k\in N}\mathbf{g}_k$. Due to the infinite divisibility of $D'$, the sum of the first $\frac{n-t'}{2m}$ terms of $\sum_{k\in N}\mathbf{g}_k$ are samples from $D'$.
        
        Then in event $E$, $\mathcal{M}_1$ is $(\epsilon, \delta - m\exp( - (n-t') / 8m )) )$ differentially private using the additive noise mechanism that adds noise from $D'$. Setting $\gamma=1/2$, Lemma \ref{lem:nonzero} gives that $E$ occurs with probability at least $1 - m\exp(-\frac{n-t'}{8m})$. Since event $E$ does not occur with probability at most $m\exp( -(n-t') / 8m )$, $\mathcal{M}_1$ is $(\epsilon, \delta- m\exp( - (n-t') / 8m )))$ differentially private except with probability $m\exp( - (n-t') / 8m )$. Then $\mathcal{M}_1$ is $(\epsilon, \delta)$-differentially private.
    
    \end{proof}

    Let $\epsilon$, $\delta$, $R_D$ and $T_\mS$ be given as described in the theorem. 
    If $\delta/d \leq m\exp( - (n-\decom-t') / 8m )$, then the $R_D$ outputs $0^d$, which is entirely independent of the input and is therefore $(\epsilon, \delta)$ differentially private.

    Otherwise, $\mathcal{M}(\mdata) = T_\mS(R_D(\data_1), \dots , R_D(\data_n))=\frac{1}{\sqrt{\decom}}\sum_{i\in[\decom]}\mS_i(\mdata+\mnoise_i)$ where $\data_i$ denotes row $i$ of $\mdata\in\mR^{n\times d}$ and $\mnoise_i\leftarrow D^{n\times d}$. Here, $\sketch$ has been decomposed such that $\sketch = \frac{1}{\sqrt{\decom}}\sum_i^\decom \sketch_i$, where \\ $\sketch_i\in\{-1,0,1\}^{m\times n}$ is a matrix with only one non-zero entry per column. 
    
    We argue $(\epsilon/\decom,\delta/\decom)$-DP based on Lemma~\ref{thm:privacy_inf_div}. Note that the non-zero entry per column in each $\sketch_i$ is not sampled uniformly at random from all possible rows; rather, for $\sketch_i$ is sampled uniformly at random from the remaining $n-i$ rows (removing the one that was already chosen to be non-zero, sampling without replacement). Therefore, we replace $n$ with $n-\decom$ 
    when applying Lemma~\ref{thm:privacy_inf_div}.

    By sequential composition, releasing $\sketch_i(\mdata+\mnoise_i)$ for all $i\in[\decom]$ satisfies $(\epsilon,\delta)$-DP, and thus by post-processing, the scaled sum of these is also $(\epsilon,\delta)$-DP.

\end{proof}

One additive mechanism to which this can easily be applied is the Gaussian mechanism. To do so, $D$ is a normal distribution with variance described in the following corollary.

\begin{corollary}\label{thm:privacynew}
    Let $\epsilon \geq 0$, $\delta \in (0,1)$, $t'<n$, $m\in[n]$, $\decom\in [n]$ and $D$ be the Gaussian distribution $\mathcal{N}(0,\sigma^2)$ with:  
    \begin{align*}
        \sigma^2 = \frac{4 \decom^3\eta^2 \ln(1.25\decom / (\delta/d - m\exp( \frac{- (n-\decom-t')}{8m} ))) m d^2}{\epsilon^2(n-\decom-t') }%\begin{cases} 
    \end{align*}
    Let $\mS\sim \mathcal{D}_{m,n,\decom}^{\algofont{sketch}}$.
    Then as long as input values are bounded above by $\eta$, tuple of algorithms $(R_D,T_\mS,A)$, as defined by equations~\ref{eq:R} and~\ref{eq:Ts} and for any $A$, is $(\epsilon,\delta)$-differentially private in the trusted computation model for differential privacy with $t'$ corrupt clients.% $C'$.
    
\end{corollary}

\begin{proof}

    We apply Theorem~\ref{thm:privacynew_general}, setting $D$ to a Gaussian distribution with the variance specified in the corollary. First note that zero-centered Gaussians are symmetric around the origin. Due to infinite divisibility of Gaussians, $\sum_{i\in[(n-s-t')/(2m)]}g_i=G$ when $\mathbf{g}_i\sim D$ is Gaussian distributed with variance \\ $\sum_{i\in[(n-s-t')/(2m)]}\sigma^2 = (n-s-t')\sigma^2/(2m)$

    Notice that the $\ell_2$-sensitivity of algorithm $T_\mS$ is $\Delta_2T_\mS=\max_{\mdata,\mdata'}\lVert\mS\mdata-\mS\mdata'\rVert_2=2\eta\sqrt{\decom}$, where $\mdata,\mdata'\in\mR^{n\times d}$ differ in a single row, and every entry of $\mdata,\mdata'$ is bounded above by $\eta$.
        
    Based on the Gaussian mechanism (Lemma \ref{lemma:gaussian_mech}), the additive noise mechanism adding Gaussian noise satisfies $(\epsilon/(\decom d), \delta/(\decom d) - m\exp( - (n-\decom -t') / 8m )) )$-DP for function $f(\mdata)=(\sketch\mdata)_j+\mathbf{n}$ with $\mathbf{n}\sim \text{Lap}(0,2\eta m^2 d/\epsilon)^m$ on input $\mdata\in\mathbb{R}^{n\times d}$, where the subscript refers to one column of $\mS\mdata$.

\end{proof}

If we would like to obtain a result for mechanisms with $\delta=0$, we can instead use $R^{\text{dense}}_D$ and use a dense sketching matrix $\sketch\in\{-1,1\}^{m\times n}$.

\begin{theorem}\label{thm:privacynew_general_dense}
Let $\epsilon \geq 0$, $\delta \in [0,1)$, $m\in[n]$, $t'<n$, $\decom\in [m]$. % and 
            % \begin{align*}
            %     \sigma^2 = \frac{4 \eta^2 \ln(1.25 / (\delta/d - m\exp( - (n-t') / 8m ))) m^2 d^2}{\epsilon^2(n-t') }
            % \end{align*}
            Let $\mS\sim \{-1,1\}^{m
            \times n}$ be a dense, uniformly sampled matrix.
            Also let $R$ add noise $g_i$ sampled from a symmetric distribution $D$ such that $\sum_{i\in[n/m-t']}g_i=G$. $G$ follows a distribution $D'$ that defines an additive noise mechanism that satisfies $(\epsilon/(m d), \delta/(m d) ) $-DP for function $f(\mdata)=(\sketch\mdata)_j+\mathbf{n}$ for $\mathbf{n}\sim D'^m$ on input $\mdata\in\mathbb{R}^{n\times d}$, where the subscript refers to one column of $\sketch\mdata$. %, assuming at least $\frac{n-t'}{2m}$ nonzero entries per row of $\ms$ corresponding to uncorrupted clients.
            Then as long as input values are bounded above by $\eta$, dense sketching mechanism $(R_D^{\text{dense}},T_\mS^{\text{dense}},A)$, as defined by equations~\ref{eq:Rdense} and~\ref{eq:Ts} and for any $A$, is $(\epsilon,\delta)$-differentially private in the LTM for differential privacy with $t'$ corrupt clients.
\end{theorem}

\begin{proof}
We again begin by proving a simplified version of the theorem first.    
\begin{lemma}\label{thm:privacy_inf_div_dense}
            Let $\epsilon \geq 0$, $\delta \in [0,1)$, $m\in[n]$, $t'<n$. 
            Let $\mS\in\{-1,0,1\}^{m\times n}$ with one non-zero entry per column and $n/m$ non-zero entries per row.
            Also let noise $g_i$ be sampled from a symmetric distribution $D$ such that $\sum_{i\in[(n)/m-t']}g_i=G$. $G$ follows a distribution $D'$ that defines an additive noise mechanism that satisfies $(\epsilon/d, \delta/d )$-DP for function $f(\mdata)=(\sketch\mdata)_j+\mathbf{n}$ with $\mathbf{n}\sim D'^m$ on input $\mdata\in\mathbb{R}^{n\times d}$, where the subscript refers to one column of $\mS\mdata$. 
            Then as long as input values are bounded above by $\eta$, tuple of algorithms $(R^{\text{dense}}_{D},T_\mS,A)$, as defined by equations~\ref{eq:Rdense} and~\ref{eq:Ts} and for any $A$, is $(\epsilon,\delta)$-differentially private in the trusted computation model for differential privacy with $t'$.
        \end{lemma}

\begin{proof}[Proof of Lemma~\ref{thm:privacy_inf_div_dense}]
    Let $\epsilon$, $\delta$, $R^{\text{dense}}_D$ and $T_\mS$ be given as described in the theorem. 
    
    We show that $\mathcal{M}(\mdata) = T_\mS(R(\data_1), \dots , R(\data_n))=\mS(\mdata+\mnoise)$ where $\data_i$ denotes row $i$ of $\mdata\in\mR^{n\times d}$ and $\mnoise\leftarrow D^{n\times d}$, is $(\epsilon,\delta)$ differentially private.

    Note that $\sketch_i(\mdata+\mnoise)$ is a d-dimensional vector, when $\sketch_i$ refers to the $i$-th row of $\sketch$.
    If $(\sketch_i(\mdata+\mnoise))_j$ is $(\epsilon/d,\delta/d)$-DP, then sequential composition (Lemma \ref{lemma:composition}) gives that $\mathcal{M}$ is $(\epsilon,\delta)$ differentially private.

    Consider $d=1$ and let mechanism $\mathcal{M}_1:\mathbb{R}^n \to \mathbb{R}^m$ be defined as $\mathcal{M}_1(\mathbf{a}) = \mathbf{S} (\mathbf{a} + \mathbf{g})$ for $\mathbf{a}\in\mathbb{R}^n$, with $\mathbf{g}\leftarrow D^n$. 
    
    $\mathbf{S}$ is sampled such that it has only one non-zero entry per column; thus, the columns of $\mS$ are orthogonal. This means that the entries of $(\mathbf{S}(\mathbf{a} + \mathbf{g}))_j$ and $(\mathbf{S}(\mathbf{a} + \mathbf{g}))_{j'}$ for $j\neq j'$, have disjoint support. Therefore, the privacy guarantee can be analyzed independently for each entry of $\mathbf{S} (\mathbf{a} + \mathbf{g})$. 

    Note that the support of $\sketch(a+g)_j$ has at least $\frac{n}{m}-t'$ uncorrupted clients, whose set of indices we denote $N$. 
    We can formulate $\mathbf{S} (\mathbf{a} + \mathbf{g})_j=(\mathbf{S} \mathbf{a})_j + (\mathbf{S}\mathbf{g})_j = (\mathbf{S} \mathbf{a})_j + \sum_{k\in[n]}\mathbf{S}_{j,k}\mathbf{g}_k = (\mathbf{S} \mathbf{a})_j + \sum_{k\in N}\mathbf{S}_{j,k}\mathbf{g}_k + \sum_{k\in [n]\setminus N}\mathbf{S}_{j,k}\mathbf{g}_k$.

     Due to the fact that $D$ is symmetric, $-1\cdot\mathbf{g}_k$ follows the same distribution as $\mathbf{g}_k$, and $\sum_{k\in N}\mathbf{S}_{j,k}\mathbf{g}_k$ follows the same distribution as $\sum_{k\in N}\mathbf{g}_k$. Due to the infinite divisibility of $D'$, the sum of the first $\frac{n}{m}-t'$ terms of $\sum_{k\in N}\mathbf{g}_k$ are samples from $D'$.
     
    Then $\mathcal{M}_1$ is $(\epsilon, \delta)$-differentially private.
\end{proof}

    Let $\epsilon$, $\delta$, $R_D^{\text{dense}}$ and $T_\mS^{\text{dense}}$ be given as described in the theorem. 
      
    We know $\mathcal{M}(\mdata) = T_\mS(R_D(\data_1), \dots , R_D(\data_n))=\frac{1}{\sqrt{m}}\sum_{i\in[m]}\mS_i(\mdata+\mnoise_i)$ where $\mnoise_i\leftarrow D^{n\times d}$. Recall from Equation~\ref{eq:Ts} that $\sketch$ has been decomposed such that $\sketch = \frac{1}{\sqrt{m}}\sum_i^m \sketch_i$, where $\sketch_i\in\{-1,0,1\}^{m\times n}$ is a matrix with only one non-zero entry per column and $n/m$ non-zero entries per row. 
    
    We argue $(\epsilon/\decom,\delta/\decom)$-DP based on Lemma~\ref{thm:privacy_inf_div_dense}.

    By sequential composition, releasing $\sketch_i(\mdata+\mnoise_i)$ for all $i\in[\decom]$ satisfies $(\epsilon,\delta)$-DP, and thus by post-processing, the scaled sum of these is also $(\epsilon,\delta)$-DP.

\end{proof}

To instantiate this second theorem, $D$ can be chosen as the distribution formulated as the difference of two samples of the gamma distribution.

\begin{corollary}\label{thm:privlap}
    Let $\epsilon \geq 0$, $t'<n$, $m\in[n]$, $\decom\in [n]$, and let $D$ be the distribution formulated as the difference of two samples of gamma distribution $\Gamma(\frac{1}{n/m-t'},b)$ with parameter  
    \begin{align*}
        b = 2\eta m^2 d/\epsilon
    \end{align*}

    Let $\mS\sim \{-1,1\}^{m\times n}$ sampled uniformly.
    Then as long as input values are bounded above by $\eta$, $(R_D^{\text{dense}},T_\mS,A)$ is $\epsilon$-differentially private in the trusted computation model for differential privacy with $t'$ corrupt clients $C'$.
    
\end{corollary}

\begin{proof}

    We apply Theorem~\ref{thm:privacynew_general_dense}, setting $D$ to the difference between two samples of the gamma distribution $\Gamma(\frac{1}{n/m-t'},2\eta m^2 d/\epsilon)$. Note that the resulting distribution is centered at 0 and symmetric. Due to infinite divisibility of Laplace distributions, $\sum_{i\in[n/m-t']}g_i=G$ when $\mathbf{g}_i\sim D$ is Laplace distributed with mean zero and scale parameter $b = 2\eta m^2 d/\epsilon$.

    Notice that the $\ell_1$-sensitivity of algorithm $T_\mS$ is \\ $\Delta_1T_\mS=\max_{\mdata,\mdata'}\lVert\mS\mdata_j-\mS\mdata'_j\rVert_1=2\eta m $, where $\mdata,\mdata'\in\mR^{n\times d}$ differ in a single row, and every entry of $\mdata,\mdata'$ is bounded above by $\eta$.

    Based on the Laplace mechanism (Lemma \ref{lemma:lap_mech}),   the additive noise mechanism adding Laplace noise satisfies $\epsilon/(dm)$-DP for function $f(\mdata)=(\sketch\mdata)_j+\mathbf{n}$ with $\mathbf{n}\sim \text{Lap}(0,2\eta m^2 d/\epsilon)^m$ on input $\mdata\in\mathbb{R}^{n\times d}$, where the subscript refers to one column of $\mS\mdata$.

\end{proof}

Corollary~\ref{cor:mpcpriv} below follows directly from Theorem~\ref{thm:privacynew_general} and Lemma~\ref{lem:mpcpriv}.
\begin{corollary}\label{cor:mpcpriv}
    Let $(R_D,T_\mS,A)$ be the tuple of algorithms above. If $T_\mS$ is computed correctly using an MPC protocol $\Pi_\mS$ with perfect security tolerating $t$ semi-honest corruptions%$\corrupts$
    , then the tuple of algorithms $(R_D,\Pi_\mS,A)$ is $(\epsilon,\delta)$-differentially private with $t'$ corrupt clients.% $\corruptc$.
\end{corollary}

\subsection{Cryptographic Assumptions and Relations to the Shuffle Model}

The LTM can be securely distributed using simple cryptographic techniques for \emph{secure multiparty computation} (MPC), such as \emph{linear secret-sharing} (LSS). MPC ensures that no adversary corrupting all but one server learns more than the computation’s output. \cite{EC:DKMMN06} first explored combining differential privacy and MPC, focusing on distributed noise generation requiring server interaction. \cite{cheu_et_al:LIPIcs.ITCS.2023.36} establish lower bounds for non-interactive multi-server mechanisms, while \cite{WWW:keller} propose an interactive MPC protocol for selection in distributed trust models.

The linear transforms can be non-interactively computed locally by servers and are also known to all participating parties, that is, once initiated, the output of the protocol is fully deterministic.
Linear transforms can be computed non-interactively by servers and are publicly known, making the protocol fully deterministic once initiated. This enables implementing the LTM with multiple central servers under the minimal assumption that at least one server remains honest. Using LSS has key advantages: clients send a single message per server proportional to the input size and sparsity $s$ of the transformation matrix $\sketch$, servers communicate only with the data analyst, and the total workload scales only with the number of servers and matrix sparsity. Since LSS supports information-theoretic security, our system remains secure even against quantum adversaries.

We require this secret sharing scheme to be information-theoretically secure. We use an additive secret sharing scheme, such that clients $C_1,\dots,C_n$ generate $k$ shares of their noisy inputs $x^1,\dots,x^n$, where $k$ is the number of servers that compute the linear transformation. 
We use $[x^i]$ to denote a secret sharing of some client $C_i$'s noisy input $x^i$ for some $i\in[n]$. For some field $\mathbb{F}$ of size $p$ and some prime number $p$, $[x^i]$ consists of shares $x^i_1,\dots,x^i_k\in\mathbb
F$ such that $\sum_{j=1}^{k}x_j^i=x^i$. To split a secret into $k$ shares, a client can sample $k-1$ random field elements $x^i_1,\dots,x^i_{k-1}$ and compute the last secret share as $x^i-\sum_{j=1}^{k-1}x_j^i$. Therefore, it also seems intuitive that an adversary that sees all but one of the shares knows nothing about the input. For every $j\in[k]$, server $S_j$ receives $\{x_j^i\}_{i\in n}$ from the respective parties. We define these shares such that the security of our distributed protocol does not rely on any computational assumption.

Since the linear transformation is public, servers can apply the transformation locally to their secret shares without any need to communicate with each other. Each server then reveals their shares of the resulting linear transformation, such that the linear sketch can be revealed in the clear. For more details on MPC based on secret sharing, see e.g.,~\cite{Cramer_Damgard_Nielsen_2015}.

Secure aggregation and the shuffle model are two instantiations of intermediate trust models for differential privacy. Secure aggregation~\cite{goryczka_comprehensive_2017,CCS:BIKMMP17,mugunthan2019smpai,apple_exposure_2021,talwar2023samplable} is a special case of the LTM, where the linear transformation applied is a sum, which is useful in federated learning. 
The shuffle model~\cite{bittau_prochlo_2017,erlingsson_amplification_2019,Cheu_2019} can be distributed among multiple servers using cryptographic protocols such as mixnets. 

However, implementing shuffles is more complex than linear transformations, requiring servers to communicate over a chain (increasing latency), computationally intensive public-key cryptography, which in turns requires computational assumptions. In contrast, LSS enables efficient, unconditional security and deterministic computation.\footnote{The linear transformations we use are random but public and sampled before the computation.}

\section{Numerical Linear Algebra in LTM}
\label{sec:NumLin}

In this section we give the utility guarantees of our mechanism. 
We start with giving generic distortion bounds relating the spectrum of the noisy matrix to the spectrum of the original non-private matrix. 
Parameterizations of this mechanism depend on the underlying class of subspace embeddings. In this utility analysis, we limit outselves to OSNAPs with $m$ and $\decom$ chosen according to~\cite{Cohen16}; other trade-offs are possible and can be found in Section~\ref{sec:prelims}. 
Applications to specific problems such as regression and low rank approximation are given towards the end of the section.

Here we present the formal utility statements for ridge regression and low rank approximation.
We begin by introducing a useful helper lemma.

\begin{lemma}
\label{lem:nicejakob}
Let $\mnoise \in \R^{n \times d}$ such that all entries in $\mnoise$ are independently sampled from a normal distribution with mean $0$ and variance at most $\sigma^2$. Further, let $V$ be a set of $d$-dimensional vectors lying in a $k$-dimensional subspace. Then with probability at least $1-\beta$ for some absolute constant $\eta$
\begin{align*}
  &\sup_{\x\in V}\|\mnoise\x\|_2^2
  \leq n\cdot \sigma^2\cdot \|\x\|^2 + \eta\cdot \sigma^2\cdot \|\x\|^2 \\
  & \cdot (\sqrt{(k+\log 1/\beta)\cdot n}+(k+\log 1/\beta)).
\end{align*}
Moreover, with probability at least $1-\beta$
\begin{align*}
  &\|\mnoise\|_F^2 \leq \eta\cdot \sigma^2\cdot n\cdot d \cdot \log 1/\beta. 
\end{align*}
\end{lemma}

\begin{proof}
    Denote by $\mnoise_i$ the $i$th row of $\mnoise$. We start by observing that $\mnoise_i^T \x$ is Gaussian distributed with mean $0$ and variance $\sigma^2\cdot \|\x\|^2$. Thus, we focus our attention on controlling $\frac{\|\mnoise\x\|^2}{\sigma^2\cdot \|\x\|^2}$, that is, we assume that $\x$ is a unit vector and that the entries of $\mnoise$ are standard normal Gaussian random variables. The overall claim then follows by rescaling.
    Consider an $\varepsilon$-net $N_{\varepsilon}$ of $V$, that is, for every $\x\in V$ with unit norm there exists a vector $\x'\in N_{\gamma}$ with $\|\x-\x'\|\leq \varepsilon$. Such nets exist with $|N_{\gamma}|\leq \exp(\eta\cdot k\log 1/\gamma)$ for some absolute constant $\eta$, see \cite{Pis99}. Suppose $\gamma = 1/4$. 
    Using concentration bounds for sums of Gaussians (see for example Lemma 1 of \cite{LaurentM00}), we then have for any $\|x'\|$ with probability $1-\beta$
    \begin{align*}
      &\mathbb{P}\big[\exists \x'\in N_{\gamma}:\|\mnoise \x'\|^2 -\mathbb{E}[\|\mnoise \x'\|^2] \\
      & \geq 2 (\sqrt{(\log N_{\gamma}+\log 1/\beta)\cdot n}+(\log N_{\gamma}+\log 1/\beta))\big] \\
      \leq & |N_{\gamma}|\cdot \exp(-(\eta\cdot k+\log 1/\beta))  \leq \beta.
    \end{align*}    
    for some absolute constant $\eta$.

    We now extend this argument to all vectors, using an argument from \cite{AroraHK06} (Lemma 4 of that reference).
    Let $\mathbf{U}$ be an orthogonal basis of $V$. Then our goal is to control 
    $\|\mnoise\x\|^2 = \x^T\mathbf{U}^T\mnoise^T\mnoise\mathbf{U} x$. Define the matrix $\mathbf{B}:=\mathbf{U}^T\mnoise^T\mnoise\mathbf{U}-n\cdot\mathbf{I}$, where $\mathbf{I}$ is the identity matrix. Note that $\mathbb{E}[\mathbf{U}^T\mnoise^T\mnoise\mathbf{U}] = n\cdot \mathbf{I}$ and that $\|\mathbf{B}\x\|$ is the deviation of $\|\mnoise\x\|$ around its expectation.
    Let $\|B\|_{op}:=\sup_{\x\in V}\|Bx\|$.
    \begin{align*}
        \|\mathbf{B}\|_{op}&=\langle \mathbf{B}\x,\x\rangle = \langle \mathbf{B}\x',\x'\rangle + \langle \mathbf{B}(\x'+\x),\x'-\x\rangle \\
        \leq& 2(\sqrt{(\log N_{\gamma}+\log 1/\beta)\cdot n}+(\log N_{\gamma}+\log 1/\beta)) \\
        &+ \|\mathbf{B}\|_{op} \cdot \|\x'+\x\|\|\x'-\x\| \\
        \leq& 2(\sqrt{(\log N_{\gamma}+\log 1/\beta)\cdot n}+(\log N_{\gamma}+\log 1/\beta)) \\
        &+ 2\gamma \cdot \|\mathbf{B}\|_{op} \\
        \leq& 2(\sqrt{(\log N_{\gamma}+\log 1/\beta)\cdot n}+(\log N_{\gamma}+\log 1/\beta)) \\
        &+ 1/2 \cdot \|\mathbf{B}\|_{op}
    \end{align*}
Rearranging implies that $\|\mathbf{B}\|_{op} \leq 4(\sqrt{(\log N_{\gamma}+\log 1/\beta)\cdot n} + (\log N_{\gamma}+\log 1/\beta))$. The first claim now follows by rescaling $\eta$.

For the second claim, we observe that $\|\mnoise\|_F^2$ has a Gaussian distribution with mean $0$ and variance $n\cdot d\cdot\sigma^2$. The same concentration inequality we applied above also implies that the probability that $\|\mnoise\|_F$ exceeds $2\sqrt{\sigma^2\cdot n\cdot d \cdot \log 1/\beta}$ is at most $1-\beta$.
\end{proof}

We will use the following spectral bounds for both linear regression and low rank approximation. We believe that there may be further applications and that the bounds themselves are therefore of independent interest.

\begin{lemma}
\label{lem:spectral_gen}
    Let $\mS\in \frac{1}{\sqrt{\decom}}\cdot \{-1,0,1\}^{m\times n}$ be an $(\alpha,\beta,m,\decom)$ OSNAP with $\mS=\frac{1}{\sqrt{s}}\sum_{i\in [\decom]}\mS_i$, where $\mS_i\in\{-1,0,1\}^{n\times m}$ has exactly one non-zero entry per row. 
    Let $\mnoise=\sum_{i\in [\decom]} \mS_i\mnoise_i$ where every matrix $\mnoise_i\in \mathbb{R}^{n\times d}$ has independent Gaussian entries $\mathcal{N}(0,\sigma^2)$. Further, let $V$ be a set of $d$-dimensional vectors lying in a $k$-dimensional subspace. Then with probability at least $1-\beta$ for some absolute constant $\eta$
\begin{align*}
  &\sup_{\x\in V}\frac{1}{s}\|\mnoise\x\|_2^2
  \leq \eta\cdot 2\frac{n}{m}\sigma^2\cdot \|\x\|^2 \cdot (\sqrt{(k+\log 1/\beta)\cdot m}\\
  & \hspace{4cm} +(k+\log 1/\beta)).
\end{align*}
and
\begin{align*}
  &\frac{1}{s}\|\mnoise\|_F^2 \leq \eta\cdot 2 n\sigma^2\cdot d \cdot \log 1/\beta
\end{align*}
\end{lemma}

\begin{proof}
We first argue that $\mnoise$ is Gaussian distributed. Each $G_i$ has independent Gaussian entries and multiplying a Gaussian with a random Rademacher does not change the distribution. Therefore, the entries of $S_iG_i$ are likewise Gaussian distributed, with mean $0$ and variance at most $(1+\gamma)\frac{n}{m}\sigma^2\leq 2\frac{n}{m}\sigma^2$ due to Lemma \ref{lem:nonzero}. Concluding, the variance of the entries of $\mnoise =\sum_{i\in [s]} S_iG_i$ is therefore at most $2s\frac{n}{m}\sigma^2$.
Applying Lemma \ref{lem:nicejakob}, we then have with probability $1-\beta$
\begin{align*}
  & \sup_{\x\in V}\|\mnoise\x\|_2^2
   \leq 2\frac{n}{m}\sigma^2\cdot \|\x\|^2  + \eta\cdot 2s\frac{n}{m}\sigma^2\cdot \|\x\|^2 \\
  & \cdot (\sqrt{(k+\log 1/\beta)\cdot m}+(k+\log 1/\beta)).
\end{align*}
% \todo{$\delta\rightarrow\beta$ in the proofs}
and
\begin{align*}
  &\|\mnoise\|_F^2 \leq \eta\cdot 2s n\sigma^2\cdot d \cdot \log 1/\beta
\end{align*}
as desired
\end{proof}

Related, though slightly weaker bounds can also be derived for the Gamma/Laplace mechanism. While it is possible to derive bounds for sparse Rademacher sketches, we focus on the dense constructions as these yield $\varepsilon$ differentially private mechanisms (as opposed to $(\varepsilon,\delta)$ differentially private mechanisms).

\begin{lemma}
\label{lem:spectral_Lap}
    Let $\mS\in \frac{1}{\sqrt{\decom}}\cdot \{-1,0,1\}^{m\times n}$ be an $(\alpha,\beta,m,m)$ OSNAP with $\mS=\frac{1}{\sqrt{m}}\sum_{i\in [\decom]}\mS_i$, where $\mS_i\in\{-1,0,1\}^{m\times n}$ has exactly one non-zero entry per row and $n/m$ non-zero entries per column. 
    Let $\mathbf{L}=\frac{1}{\sqrt{m}}\sum_{i\in [m]} \mS_i\mathbf{\Gamma}_i$ where every matrix $\mathbf{\Gamma}_i\in \mathbb{R}^{n\times d}$ has independent entries $X_i-Y_i$, where $X_i,Y_i\sim\Gamma(m/n,b)$. Then with probability at least $1-\beta$ for some absolute constant $\eta$
% \todo{$\delta\rightarrow\beta$ in the proofs}
and
\begin{align*}
  &\|\mathbf{L}\|_F^2 \leq \eta d m^2 b^2 \log^2 \frac{m}{\beta}
\end{align*}
\end{lemma}
\begin{proof}
    We have $\|\mathbf{L}\|_F\leq \sqrt{m}\cdot \|\max_{i\in [m]} \mS_i\mathbf{\Gamma}_i\|_F$, thus it is sufficient to give a high probability bound on $\|\mS_i\mathbf{\Gamma}_i\|_F$.
    
    The entries of $\mS_i\mathbf{\Gamma}_i$ are Laplace distributed entries with mean $0$ and scale parameter $b$ by infinite divisibility of a Laplace random variable into $\Gamma$ distributed random variables. Thus the expected value of $\|\mS_i\mathbf{\Gamma}_i\|_F^2$ is $md\cdot 2b^2$. 

    To prove high concentration, let us consider the moments of the Laplace distribution. For a Laplace random variable $L$ with mean $0$ and scale $b$, we have $\mathbb{E}[L^{2\ell}] = (2\ell)! b^{2\ell}$. Then due to Hoelder's inequality
    \begin{align*}
      \mathbb{E}\left[\|\mS_i\mathbf{\Gamma}_i\|_F^{2\ell}\right] =& \mathbb{E}\left[\left(\sum_{j=1}^{md} L_j^2\right)^\ell \right]\\
      \leq& \mathbb{E}\left[\left(\sum_{j=1}^{md} L_i^{2\ell}\right)\right]\cdot \left(\sum_{j=1}^{md} 1^{\frac{\ell}{\ell-1}}\right)^{\ell-1}\\
    = &(2\ell)! b^{2\ell} \cdot (md)^{\ell}
    \end{align*}
    Applying Markov's inequality, we then have for any $\gamma>1$
    \begin{align*}
        &\mathbb{P}\left[\|\mS_i\mathbf{\Gamma}_i\|_F^{2} > \gamma\cdot md\cdot 2b^2\right] \\
        =& \mathbb{P}\left[\|\mS_i\mathbf{\Gamma}_i\|_F^{2\ell} > \left(\gamma\cdot md\cdot 2b^2\right)^{\ell}\right]\\
        \leq & \frac{(2\ell)!b^{2\ell}(md)^\ell}{\left(\gamma\cdot md\cdot 2b^2\right)^{\ell}} = \frac{(2\ell)!}{(2\gamma)^\ell}\\
        \leq& \frac{(2\ell)^{2\ell}}{(2\gamma e^2)^\ell}\sqrt{6\pi \ell}
    \end{align*}
    For $\gamma >4\ell^2$ and $\ell = \log \frac{m}{\beta}$, this term is at most $\beta/m$. We can therefore conclude that with probability $1-\beta$, $\max_{i\in [m]} \mS_i\mathbf{\Gamma}_i\leq 4\log \frac{m}{\beta}$ and thus $\|\mathbf{L}\|_F^2\leq \eta m^2 d b^2 \log^2 \frac{m}{\beta}$, for some absolute constant $\eta$.
\end{proof}

We first begin with our utility results for low rank approximation. Using the Gaussian mechanism, we obtain the following guarantee.
    
\begin{theorem}\label{thrm:utilitykrank}
    Let $\epsilon \geq 0$, $\delta \in (0,1)$, $t'<n/2$.
    Let $\mS\sim \mathcal{D}_{m,n}^{\algofont{sketch}}$. Define $A_{k}$ as an algorithm that computes $\text{argmin}_{\mx \text{ rank } k,\in \mathbb{R}^{n\times d}}\lVert\mdata - \mdata \mx \mx^T \rVert_F^2$. 
    An instantiation of the trusted computation model $(R_D,T_\sketch,A_{k})$ specified in Equations~\ref{eq:R} and \ref{eq:Ts}, using noise distribution $D \sim \mathcal{N}(0,\sigma^2)$ according to Corollary \ref{thm:privacynew}, computes an orthogonal projection $\mx'$ that with probability $1-\beta$ achieves a multiplicative error $(1+O(\smult))$ and additive error $O\left(\frac{kd^{3/2}}{\smult^{3}} \log^{3/2} (1/\beta)\varepsilon^{-1}\log^{1/2}(1/\delta)\right)$ for the low rank approximation problem.% additive and multiplicative errors:
\end{theorem}

\begin{proof}
    We will use the following inequality. The applications to vector and matrix norms are straightforward corollaries.
    \begin{lemma}[Generalized Triangle Inequality \cite{BecchettiBC0S19}]\label{lemma:magic}
    For any two real numbers $a,b$ and any $\alpha>0$
    $$|a^2 - b^2|\leq \alpha\cdot a^2 + \left(1+\frac{1}{\alpha}\right)\cdot (a-b)^2.$$
    \end{lemma}
    
    Finally, we require the following consequence of subspace embedding properties. 
    \begin{lemma}[Paraphrased from Definition 1, Lemma 11, and Theorem 12 \cite{CohenEMMP15}]
    \label{lem:subspaceembedding}
        $\mS$ be drawn from an OSNAP distribution. Then for any matrix $A\in \mathbb{R}^{n\times d}$ and any rank $k$ orthogonal matrix $\mx\in \mathbb{R}^{d\times k}$, we have
        $$\|\mS\mdata-\mS\mdata\mx\mx^T\|_F = (1\pm \alpha)\cdot \|\mdata-\mdata\mx\mx^T\|_F.$$
    \end{lemma}
    
    Without noise, the low rank approximation outputs $\mxopt=\text{argmin}_{\mx \text{ rank } k,\in \mathbb{R}^{n\times d}}\lVert\mdata - \mdata \mx \mx^T \rVert_F^2$, and after sketching with Gaussian noise of variance $\sigma^2$ outputs $\mxest=\text{argmin}_{\mx\text{ rank } k,\in \mathbb{R}^{n\times d}}\lVert (\mdata+\mnoise) -(\mdata+\mnoise) \mx\mx^T \rVert_F^2$. 
    Let $\mxest$ be the matrix returned by the mechanism. We have
        \begin{align*}
            \lVert &  \mdata-\mdata\mxest\mxest^T \rVert_F \\
            & \leq (1+\smult) \sqrt{\frac{1}{s}}\left\lVert \sum_{i\in[\decom]}\mS_i\mdata-\mS_i\mdata\mxest\mxest^T \right\rVert_F \\
            & \leq (1+\smult)\sqrt{\frac{1}{s}}\\
            & \left\lVert\sum_{i\in[\decom]}\mS_i(\mdata+\mnoise_i) - \mS_i(\mdata+\mnoise_i)\mxest\mxest^T \right\rVert_F \\
            & + (1+\smult)\sqrt{\frac{1}{s}}\left\lVert \sum_{i\in[\decom]}\mS_i\mnoise_i-\mS_i\mnoise_i \mxest\mxest^T\right\rVert_F
        \end{align*}
        where the first inequality follows from the subspace embedding property. By optimality of $\mxest$ for the low rank approximation problem on $ \sum_{i\in[\decom]}\mS_i(\mdata+\mnoise_i) $, we then have
        \begin{align*}
             \lVert &  \mdata-\mdata\mxest\mxest^T \rVert_F  \\
             & \leq (1+\smult)\sqrt{\frac{1}{s}} \\
             & ~~\cdot\left\lVert \sum_{i\in[\decom]}\mS_i(\mdata+\mnoise_i) - \mS_i(\mdata+\mnoise_i)\mxopt\mxopt^T \right\rVert_F \\
            & ~~+ (1+\smult)\sqrt{\frac{1}{s}}\left\lVert \sum_{i\in[\decom]}\mS_i\mnoise_i-\mS_i\mnoise_i \mxest\mxest^T\right\rVert_F  \\
            & \leq (1+\smult)^2\left\lVert \mdata - \mdata\mxopt\mxopt^T\right \rVert_F \\
            &~~~~~~~ + (1+\smult)\sqrt{\frac{1}{s}}\left\lVert \sum_{i\in[\decom]}\mS_i\mnoise_i-\mS_i\mnoise_i \mxopt\mxopt^T\right\rVert_F  \\
            &~~+ (1+\smult)\sqrt{\frac{1}{s}}\left\lVert \sum_{i\in[\decom]}\mS_i\mnoise_i-\mS_i\mnoise_i \mxest\mxest^T\right\rVert_F \\
        \end{align*}
    where the inequality follows from Lemma \ref{lem:subspaceembedding}.
        Now notice that $\mathbb{I}-\mxopt\mxopt^T$ and $\mathbb{I}-\mxest\mxest^T$ are orthogonal projections, multiplying by which cannot increase the Frobenius norm of a matrix. Therefore, for $\mx=\mxest$ or $\mx=\mxopt$:
        \begin{align*}
            & \sqrt{\frac{1}{s}}\left\lVert \sum_{i\in[\decom]} \mS_i\mnoise_i-\mS_i\mnoise_i \mx\mx^T\right\rVert_F \\
            & =\sqrt{\frac{1}{s}}\left\lVert \sum_{i\in[\decom]} \mS_i\mnoise_i(\mathbb{I}- \mx\mx^T)\right\rVert_F  \leq\sqrt{\frac{1}{s}}\left\lVert \sum_{i\in[\decom]} \mS_i\mnoise_i\right\rVert_F
        \end{align*}
    
        Then:
        \begin{align}
        \nonumber
            &\lVert   \mdata-\mdata\mxest\mxest^T \rVert_F  \\
             \nonumber
            &\leq (1+\smult)^2\lVert \mdata - \mdata\mxopt\mxopt^T \rVert_F \\
            \label{eq:lowrank}
            & + (1+\smult)2\frac{1}{s}\left\lVert \sum_{i\in[\decom]} \mS_i\mnoise_i\right\rVert_F
        \end{align}
    Using Lemma \ref{lem:spectral_gen}, we have $\frac{1}{s}\left\lVert \sum_{i\in[\decom]} \mS_i\mnoise_i\right\rVert_F^2\leq 2n\cdot \sigma^2\cdot d\log(1/\beta)$. To simplify the error, observe that for the sketching matrix by \cite{Cohen16} (see Table \ref{tab:ose}), we can set $\decom=\frac{\log k}{\smult}$ and $m=\frac{k\log k}{\smult^2}\log(1/\beta)$. Using the bound on the variance from Theorem \ref{thm:privacynew} and plugging this in above, we can simplify the additive and multiplicative errors.
    
        \begin{align*}
             \lVert\mdata & - \mdata \mxest \mxest^T \rVert_F^2 \\
            & \leq (1+\smult)^2\lVert \mdata - \mdata\mxopt\mxopt^T \rVert_F^2 \\
            &\phantom{abddddd}+ \eta(1+\smult)\sqrt{\sigma^2\cdot n\cdot d\cdot \log \frac{1}{\Tilde{\beta}}}   \\
            & \leq (1+3\smult)\lVert \mdata - \mdata\mxopt\mxopt^T \rVert_F^2 \\
            &\phantom{abddddd}+ \Tilde{O}\left(\frac{kd^{3/2}}{\smult^3} \log^{3/2} (1/\beta)\varepsilon^{-1}\log^{1/2}(1/\delta)\right)
        \end{align*}
    The stated bound follows by rescaling $\alpha_\mS$ by a factor of $3$.
\end{proof}

Using dense Rademacher sketches and the Gamma/Laplace mechanism, we instead obtain the following guarantee.
\begin{theorem}\label{thrm:utilitykrank_lap}
    Let $\epsilon \geq 0$, $\delta \in (0,1)$, $t'<n/2$.
    Let $\mS\sim \mathcal{D}_{m,n}^{\algofont{sketch}}$. Define $A_{k}$ as an algorithm that computes $\text{argmin}_{\mx \text{ rank } k,\in \mathbb{R}^{n\times d}}\lVert\mdata - \mdata \mx \mx^T \rVert_F^2$. 
    An instantiation of the trusted computation model \\$(R_D^{\text{dense}},T_\sketch^{\text{dense}},A_{k})$ specified in Equations~\ref{eq:Rdense} and \ref{eq:Ts}, using noise distribution $D$ as the difference between two independent $\Gamma$ random variables parameterized by Corollary \ref{thm:privlap}, computes an orthogonal projection $\mx'$ that with probability $1-\beta$ with multiplicative error $(1+O(\smult))$ and additive error $O(\frac{k^{3}d^{3/2}}{\alpha_\mS^6} \varepsilon^{-1} \log \frac{m}{\beta})$.
\end{theorem}

\begin{proof}
    The proof follows the same line of reasoning as Theorem \ref{thrm:utilitykrank}. The main difference is that upon reaching Equation \ref{eq:lowrank}, we must bound $\frac{1}{m}\left\lVert \sum_{i\in[m]} \mS_i\mnoise_i\right\rVert_F^2$ via Lemma \ref{lem:spectral_Lap} rather than using Lemma \ref{lem:spectral_gen}. Assuming the semi-trusted corruption model with at most $t$ corruptions, each entry of $\mnoise_i$ consists of $X_j-Y_j$ where $X_j,Y_i\sim\Gamma(1/(n-t),b)$. Consequently, the entries of $\sum_{i\in[m]} \mS_i\mnoise_i$ are a sum of $m\cdot \lceil\frac{n}{n-t}\rceil$ matrices where each entry is either $0$ or a Laplacian distributed random variable with mean $0$ and scale $b$. Using Lemma \ref{lem:spectral_Lap}, we therefore have with probability $1-\beta$.
    \begin{align*}
      &\frac{1}{m}\left\lVert \sum_{i\in[m]} \mS_i\mnoise_i\right\rVert_F^2 \\
      &\in O(1)\cdot m^2 \cdot \frac{n}{n-t} d b^2 \log^2 \frac{m \cdot \frac{n}{n-t}}{\beta} \\
      &\in O(1)\cdot m^2 \cdot d b^2 \log^2 \frac{m }{\beta},
    \end{align*}
    where we used $\frac{n}{n-t}\leq 2$ by assumption.
    Since $\mS$ is a dense Rademacher matrix, we have $m = \frac{k+\log \beta^{-1}}{\alpha_\mS^2}$ \cite{Achlioptas03,Woodruff14} and $s= m$, and $b=2\eta m^2 d/\epsilon$ following Corollary \ref{thm:privlap}.

    Thus we obtain, with probability $1-\beta$, a $(1+\alpha_\mS)$ multiplicative approximation using the same calculation as in Theorem \ref{thrm:utilitykrank} and a suitable rescaling, and an additive $O(\frac{k^{3}d^{3/2}}{\alpha_\mS^6} \varepsilon^{-1} \log^2 \frac{k}{\alpha_\mS\beta})$ approximation.
\end{proof}

For ridge regression, we achieve the following bound.

\begin{theorem}\label{thrm:linreg}
    Let $\epsilon \geq 0$, $\delta \in (0,1)$, $t'<n$ and $\sigma$ be chosen as described in Theorem~\ref{thm:privacynew}. 
    Let $\mS\sim \mathcal{D}_{m,n}^{\algofont{sketch}}$. Define $A_{\lambda}$ that performs linear regression on the sketched noisy inputs, outputting $\text{argmin}_{\x}\lVert \mdata \x -\predictor \rVert^2 +\lambda \lVert \x \rVert^2$, where the error is measured by $\text{min}_{\x}\lVert \mdata \x -\predictor \rVert^2 +\lambda \lVert \x \rVert^2$. An instantiation of the LTM $(R_D,T_\sketch,A_{\lambda})$ specified in Equations~\ref{eq:R} and \ref{eq:Ts}, using noise distribution $D \sim \mathcal{N}(0,\sigma^2)$ according to Corollary \ref{thm:privacynew}, which tolerates $t'$ corrupted clients, computes linear regression parameters for any $\smult>0$ with probability $1-\beta$ with multiplicative error $1+\Tilde{O}\left(\alpha_S + p+ p^2\right)$ and additive error $\tilde{O}\left(p+p^2\right)$, where $p=\frac{d^3\alpha_S^{-5} \log^5 \frac{1}{\beta}\cdot \varepsilon^{-2}\log \frac{1}{\delta}}{\lambda}$.
\end{theorem}

\begin{proof}
            Let $\epsilon$, $\delta$, $R_D$ and $T_\mS$ be given as in the theorem. We consider the output $A_{reg}(T_\mS(R_D(\data_1), \dots , R_D(\data_n)))$. 
            
            If $ n\leq 8m \ln(dm/\delta)+t'$, then the $R_D$ outputs $0^d$, then the optimal $\xest=0^d$, leading to error $\lVert \mdata\xest-\predictor \rVert_2^2 + \lambda\lVert\xest\rVert_2^2=\lVert\predictor\rVert^2\leq \eta' n\leq \eta' m \ln(dm/\delta)+t'$.
            
            Otherwise, $\mathcal{M}(\mdata) = T_\mS(R_D(\data_1), \dots , R_D(\data_n))=\frac{1}{\sqrt{\decom}}\sum_{i\in[\decom]}\mS_i(\mdata+\mnoise_i)$ where $\data_i$ denotes row $i$ of $\mdata\in\mR^{n\times d}$ and $\mnoise\leftarrow\mathcal{N}(0,\sigma^2)^{n\times d}$, is $(\epsilon,\delta)$ differentially private.
    
    To streamline the presentation, we give the analysis without corrupted clients. Adding corrupted clients merely changes the analysis along the same lines as Theorem \ref{thrm:utilitykrank}.
    
    We first the control the terms $ \|\sum_{i\in [s]}\sketch_i (\mnoise_i \x'-\noise_i)\|^2$ and \\ $ \|\sum_{i\in [s]}\sketch_i (\mnoise_i \xopt-\noise_i)\|^2$.
    Using Lemma \ref{lem:spectral_gen} with an added coordinate of $-1$ to both $\x'$ and $\xopt$, we get
    \begin{equation}
    \frac{1}{s}\|\sum_{i\in [s]}\sketch_i (\mnoise_i \x-\noise_i)\|^2 \leq \eta\cdot \sigma^2nd\log \frac{1}{\beta} (\|\x\|^2+1)
        \label{eq:specbound}
    \end{equation}
    for either $\x = \xopt$ or $\x=\x'$ and for a sufficiently large constant $\eta.$
    Then we have using Lemma \ref{lem:subspaceembedding}
            \begin{align}
                &\|\mdata \x'-\predictor\|^2 + \lambda \|x'\|^2  \nonumber \\
                & \leq (1+\alpha_S)\|\sketch(\mdata \x'-\predictor)\|^2 + \lambda \|x'\|^2 \nonumber \\
               = & (1+\alpha_S)\frac{1}{s} \|\sum_{i\in [s]}\sketch_i(\mdata \x'-\predictor)\|^2 + \lambda \|x'\|^2 \nonumber \\
               =& \lambda \|x'\|^2 + (1+\alpha_S)\frac{1}{s} \nonumber \\
               & \cdot \|\sum_{i\in [s]}\sketch_i((\mdata  +\mnoise_i- \mnoise_i)\x'-\predictor +\noise_i -\noise_i)\|^2 \nonumber \\
               % &  + \nonumber\\
               \leq &(1+\alpha_S)\frac{1}{s}  \bigg((1+\alpha_S)\|\sum_{i\in [s]}\sketch_i((\mdata  +\mnoise_i)\x'-(\predictor +\noise_i))\|^2 \nonumber\\
               &+ \left(1 +\frac{1}{\alpha_S}\right) \|\sum_{i\in [s]}\sketch_i (\mnoise_i \x'-\noise_i)\|^2\bigg)  + \lambda \|\x'\|^2 \nonumber\\
               \leq & (1+\alpha_S)^2\cdot \Bigg(\frac{1}{s}\|\sum_{i\in [s]} S_i((\mdata +\mnoise_i)\x' - (\predictor+\noise_i)\|^2 \\
               & + \lambda\cdot \|\x'\|^2\Bigg)  + 2\cdot \left(1 +\frac{1}{\alpha_S}\right)\eta \sigma^2 nd\log \frac{1}{\beta}(\|\x'\|^2 +1)
               \label{eq:temp1}
            \end{align}
    where the second to last inequality follows by applying Lemma \ref{lemma:magic} and the final inequality follows from Equation \ref{eq:specbound}. By optimality of $\x'$ for the instance 
            $\frac{1}{s}\|\sum_{i\in [s]}\sketch_i(\mdata \x' +\mnoise_i-(\predictor +\noise_i))\|^2 + \lambda \|x'\|^2$, we then have
    \begin{align*}
        &\frac{1}{s}\|\sum_{i\in [s]} S_i((\mdata + \mnoise_i)\x' - (\predictor+\noise_i)\|^2 + \lambda\cdot \|\x'\| \\ & \leq \frac{1}{s}\|\sum_{i\in [s]} S_i((\mdata + \mnoise_i)\xopt - (\predictor+\noise_i)\|^2 + \lambda\cdot \|\xopt\|^2
    \end{align*}
    which likewise implies
    \begin{align*}
         \|\x'\| \leq \frac{1}{\lambda}\cdot\Bigg( & \frac{1}{s}\|\sum_{i\in [s]} S_i((\mdata + \mnoise_i)\xopt - (\predictor+\noise_i))\|^2  + \lambda\cdot \|\xopt\|^2\Bigg).
    \end{align*}
    
    Insertion both bounds back into Equation \ref{eq:temp1}, we obtain
    
            \begin{align}
                &\|\mdata \x'-\predictor\|^2 + \lambda \|x'\|^2 \nonumber\\
               & \leq (1+\alpha_S)^2 \nonumber\\
               & \left(\frac{1}{s}\|\sum_{i\in [s]}\sketch_i((\mdata  +\mnoise_i)\xopt-(\predictor +\noise_i))\|^2 + \lambda \|\xopt\|^2\right) \nonumber \\
               & + 2\left(1 +\frac{1}{\alpha_S}\right)\eta \sigma^2 nd\log \frac{1}{\beta} + 2\left(1 +\frac{1}{\alpha_S}\right)\frac{\eta \sigma^2 nd\log \frac{1}{\beta}}{\lambda}\\
               & \cdot\Bigg(\frac{1}{s}\|\sum_{i\in [s]} S_i((\mdata +\mnoise_i)\xopt - (\predictor+\noise_i))\|^2  + \lambda\cdot \|\xopt\|^2\Bigg) \label{eq:temp2}
            \end{align}
    We now turn our attention to $\|\sum_{i\in [s]}\sketch_i((\mdata  +\mnoise_i)\xopt-(\predictor +\noise_i))\|^2$. Using Lemma \ref{lemma:magic} and Equation \ref{eq:specbound}, we have
    \begin{align*}
        \frac{1}{s}&\|\sum_{i\in [s]}\sketch_i((\mdata  +\mnoise_i)\xopt-(\predictor +\noise_i))\|^2 + \lambda\cdot \|\xopt\|^2\\
        \leq & (1+\alpha_S)\frac{1}{s}\|\sum_{i\in [s]}\sketch_i(\mdata \xopt -\predictor)\|^2 \\
        &+  \left(1+\frac{1}{\alpha_S}\right)\frac{1}{s}\|\sum_{i\in [s]}\sketch_i(\mnoise_i \xopt -\noise_i)\|^2  + \lambda\cdot \|\xopt\|^2\\
        \leq & (1+\alpha_S)^2\|\mdata \xopt -\predictor\|^2 \\
        &+  \left(1+\frac{1}{\alpha_S}\right) \eta \sigma^2 nd\log \frac{1}{\beta}(\|\xopt\|^2 +1)  + \lambda \|\xopt\|^2\\
        \leq & (1+\alpha_S)^2\left(\|\mdata \xopt -\predictor\|^2 + \lambda\|\xopt\|^2\right) \\
        & +\left(1+\frac{1}{\alpha_S}\right) \frac{\eta \sigma^2 nd\log \frac{1}{\beta}}{\lambda}(\lambda\|\xopt\|^2 +\lambda) \\
        \leq & \left((1+\alpha_S)^2 + \left(1+\frac{1}{\alpha_S}\right) \frac{\eta \sigma^2 nd\log \frac{1}{\beta}}{\lambda}\right) \\
        & \left(\|\mdata \xopt -\predictor\|^2 + \lambda \cdot \|\xopt\|^2\right) \\
        & +\left(1+\frac{1}{\alpha_S}\right) \eta \sigma^2 nd\log \frac{1}{\beta}
    \end{align*}
    
    Inserting this into Equation \ref{eq:temp2} and collecting all the terms, we obtain
            \begin{align}
                &\|\mdata \x'-\predictor\|^2 + \lambda \|x'\|^2 \nonumber\\
               \leq & \Bigg( \frac{32}{\alpha_S}\left(\frac{\eta \sigma^2 nd\log \frac{1}{\beta}}{\lambda} + \frac{2}{\alpha_S}\left(\frac{\eta \sigma^2 nd\log \frac{1}{\beta}}{\lambda}\right)^2\right) \\
               & + (1+\alpha_S)^4\Bigg) \left(\|\mdata\xopt-\predictor\|^2 + \lambda \|\xopt\|^2\right) \nonumber \\
               & + \frac{32}{\alpha_S}\left(\frac{\eta \sigma^2 nd\log \frac{1}{\beta}}{\lambda} + \frac{2}{\alpha_S}\left(\frac{\eta \sigma^2 nd\log \frac{1}{\beta}}{\lambda}\right)^2\right) \nonumber
            \end{align}        
    
    By our choice of $\sigma^2$ from Theorem \ref{thm:privacynew} and using the sketching matrix of \cite{Cohen16} (see Table \ref{tab:ose}), we have
    $\sigma^2\cdot n \in O(d^2\log^4 d  \cdot \alpha_S^{-4} \cdot \log^4 \frac{1}{\beta} \varepsilon^{-2} \cdot \log\frac{1}{\delta})$. Thus, we have
    
         \begin{align*}
            &\|\mdata \x'-\predictor\|^2 + \lambda \|x'\|^2 
           \leq  (\|\mdata  \xopt-\predictor \|^2 + \lambda \|\xopt\|^2) \cdot\Bigg(1+15\alpha_S \\
           & + \tilde{O}\Bigg(\frac{d^3\alpha_S^{-5} \log^5 \frac{1}{\beta}\cdot \varepsilon^{-2}\log \frac{1}{\delta}}{\lambda}  + \frac{d^6\alpha_S^{-10} \log^{10} \frac{1}{\beta}\cdot \varepsilon^{-4}\log^2 \frac{1}{\delta}}{\lambda^2}\Bigg)\Bigg)\\
           &+ \tilde{O}\Bigg(\frac{d^3\alpha_S^{-5} \log^5 \frac{1}{\beta}\cdot \varepsilon^{-2}\log \frac{1}{\delta}}{\lambda}  + \frac{d^6\alpha_S^{-10} \log^{10} \frac{1}{\beta}\cdot \varepsilon^{-4}\log^2 \frac{1}{\delta}}{\lambda^2}\Bigg)
        \end{align*}

\end{proof}

Note that we can get a constant multiplicative error as long as the regularization factor $\lambda$ depends on a sufficiently large polynomial in $d$, while being independent of $n$. Generally such a relationship still has good generalization properties when training a regression model. If $\lambda$ is sufficiently large, one should choose $\alpha_S$ to be as small as possible, minimizing the tradeoff between $\alpha_S$ and $\frac{\alpha_S^{-5}}{\lambda}$. If $\lambda$ is not sufficiently large, one should choose $\alpha_S = 1/3$ such that $\sketch$ is an oblivious subspace embedding, but the target dimension does not have a prohibitively large dependency on the sketch distortion.

\section{Experimental Evaluation}
\label{sec:experiments}

Section \ref{experiments_appendix_runtime} provides the running time experiments and in Section \ref{experiments_appendix_interpolation} we discuss our utility investigation both for ridge regression and low-rank approximation.

\subsection{Running Time of LTM} \label{experiments_appendix_runtime}

The only overhead of the LTM over the local model is the execution of the linear transform in MPC, and comes from distributing $\mdata + \mnoise_i \in \R^{n \times d}$ among the servers for $i \in [s]$ and then performing matrix multiplications $\sketch_i(\mdata + \mnoise_i)$ in MPC, where $\sketch_1,\dots,\sketch_s$ is a decomposition of some $\sketch \in \R^{m \times n}$ chosen according to \cite{nelson2013osnap,Cohen16} with sparsity $s$.
We implemented our approach in MP-SPDZ~\cite{DBLP:conf/ccs/Keller20}, a popular and easy to use MPC framework, for multiple combinations of $n$ and $s$, using $S$ servers from Amazon Web Services t3.large instances.
We use additive secret sharing over the ring of integers modulo $2^{64}$, which results in no communication between servers, and semi-honest security against an adversary that corrupts all but one servers.

We fix the dimensionality $d=10$ and $m=50$, and vary the number of clients $n \in \{ 100000,250000,500000,750000,1000000 \}$, the sparsity $s \in \{ 1, 10, 20, 30, 40, 50 \}$ of the sketch (which dictates how many linear transformations we need to apply) and the number of servers $S \in \{ 2,3 \}$. The data matrix $\mdata$ is generated at random such that all entries are smaller than $2^{32}$ and the decomposition $\sketch_1,\dots,\sketch_s$ of some sketch $\sketch$ is generated according to our mechanism. Every parameter setting is then evaluated by running the protocol $10$ times and averaging over the running times. Table~\ref{tab:mp-spdz-3S} shows the runtime for each server and the total communication load for one particular combination of parameters, which includes three servers: Even with one million clients, the computation on each servers lasts less than $2$ seconds. We conclude that the use of MPC in the LTM does not hinder the computational efficiency of our proposed mechanisms, showing that they can easily be used in practice. Table~\ref{tab:mp-spdz-2S} reports communication and time required when running the protocol with two servers, and Table~\ref{tab:mp-spdz-vary-s} reports time and communication requirements when varying $\decom$.

\begin{table}[h]
    \caption{Computation cost $T_{\text{MPC}}$ and communication cost $C_{\text{MPC}}$ of the LTM using MPC for varying number of clients $n$. Here $m=100$, $d=10$, $s=1$ and we are working with $S=3$ servers.}
    \label{tab:mp-spdz-3S}
    \centering
        \begin{tabular}{|c|c|c|}
            \hline
            $n$ & $T_{\text{MPC}}$ (sec) & $C_{\text{MPC}}$ (MB) \\
            \hline
            $100000$ & $0.172 \pm 0.019$ & $24.024$  \\
            $250000$ & $0.457 \pm 0.058$ & $60.024$  \\
            $500000$ & $0.861 \pm 0.102$ & $120.024$  \\
            $750000$ & $1.219 \pm 0.128$ & $180.024$  \\
            $1000000$ & $1.641 \pm 0.165$ & $240.024$  \\
            \hline
        \end{tabular}
\end{table}

\begin{table}[h]
    \caption{Computation cost $T_{\text{MPC}}$ and communication cost $C_{\text{MPC}}$ of the LTM using MPC %compared to the central model 
    for varying number of clients $n$. Here $m=100$, $d=10$, $s=1$ and we are working with $S=2$ servers.}
    \label{tab:mp-spdz-2S}
    \centering
        \begin{tabular}{|c|c|c|}%c|}
            \hline
            $n$ & $T_{\text{MPC}}$ (sec) & $C_{\text{MPC}}$ (MB)  %& $C_{\text{Central}}$ (MB) 
            \\
            \hline
            $100000$ & $0.177 \pm 0.028$ & $16.016$ %& $8$ 
            \\
            $250000$ & $0.453 \pm 0.036$ & $40.016$ %& $20$ 
            \\
            $500000$ & $0.824 \pm 0.101$ & $80.016$ %& $40$ 
            \\
            $750000$ & $1.240 \pm 0.157$ & $120.016$ %& $60$ 
            \\
            $1000000$ & $1.683 \pm 0.130$ & $160.016$ %& $80$ 
            \\
            \hline
        \end{tabular}
\end{table}

\begin{table}[h]
    \caption{Computation cost $T_{\text{MPC}}$ and communication cost $C_{\text{MPC}}$ of the LTM using MPC 
    for varying sparsity $s$. Here $m=100$, $d=10$, $n=500000$ and we are working with $S=3$ servers.}
    \label{tab:mp-spdz-vary-s}
    \centering
        \begin{tabular}{|c|c|c|}%c|}
            \hline
            $s$ & $T_{\text{MPC}}$ (sec) & $C_{\text{MPC}}$ (MB) 
            %& $C_{\text{Central}}$ (MB) 
            \\
            \hline
            $1$ & $0.851 \pm 0.099$ & $80.016$ 
            %& $40$ 
            \\
            $10$ & $6.863 \pm 0.388$ & $800.016$ %& $40$ 
            \\
            $20$ & $13.002 \pm 0.445$ & $1600.016$ %& $40$ 
            \\
            $30$ & $20.356 \pm 1.171$ & $2400.016$ %& $40$ 
            \\
            $40$ & $28.023 \pm 0.660$ & $3200.016$ %& $40$ 
            \\
            $50$ & $33.708 \pm 1.384$ & $4000.016$ %& $40$ 
            \\
            \hline
        \end{tabular}
\end{table}

\subsection{Interpolation between Local and Central DP}\label{experiments_appendix_interpolation}

\begin{figure*}%[tb]
	\centering
	\begin{subfigure}%{0.49\linewidth}
        \centering
            \includegraphics[width=0.4\linewidth]{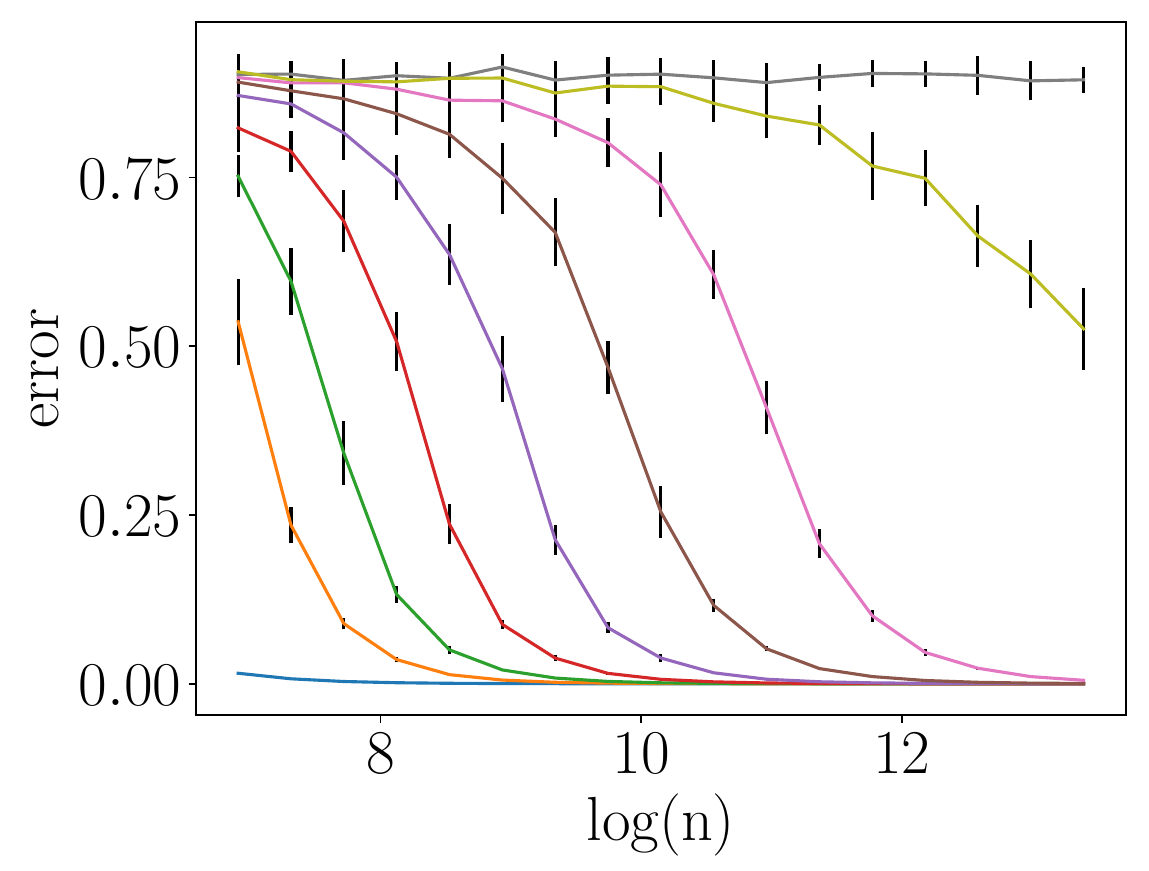}
            \includegraphics[width=0.4\linewidth]{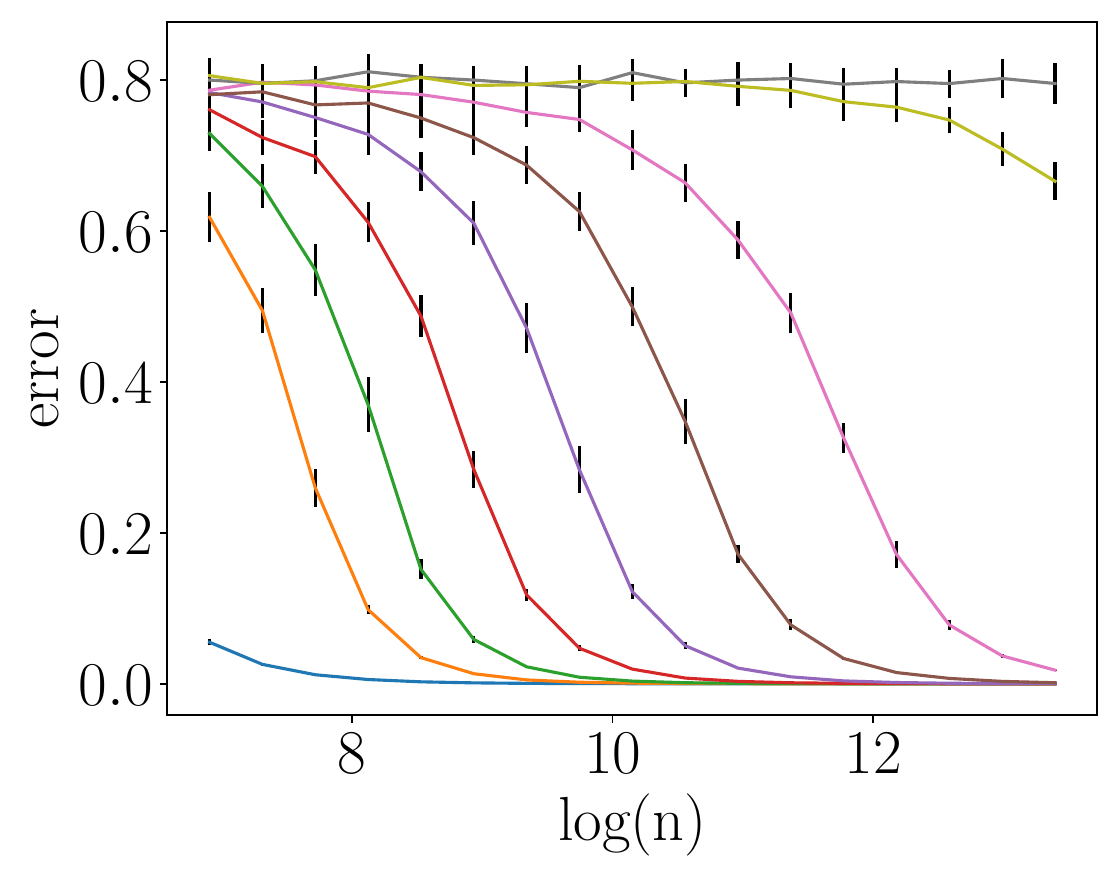}\\
            \includegraphics[width=0.4\linewidth]{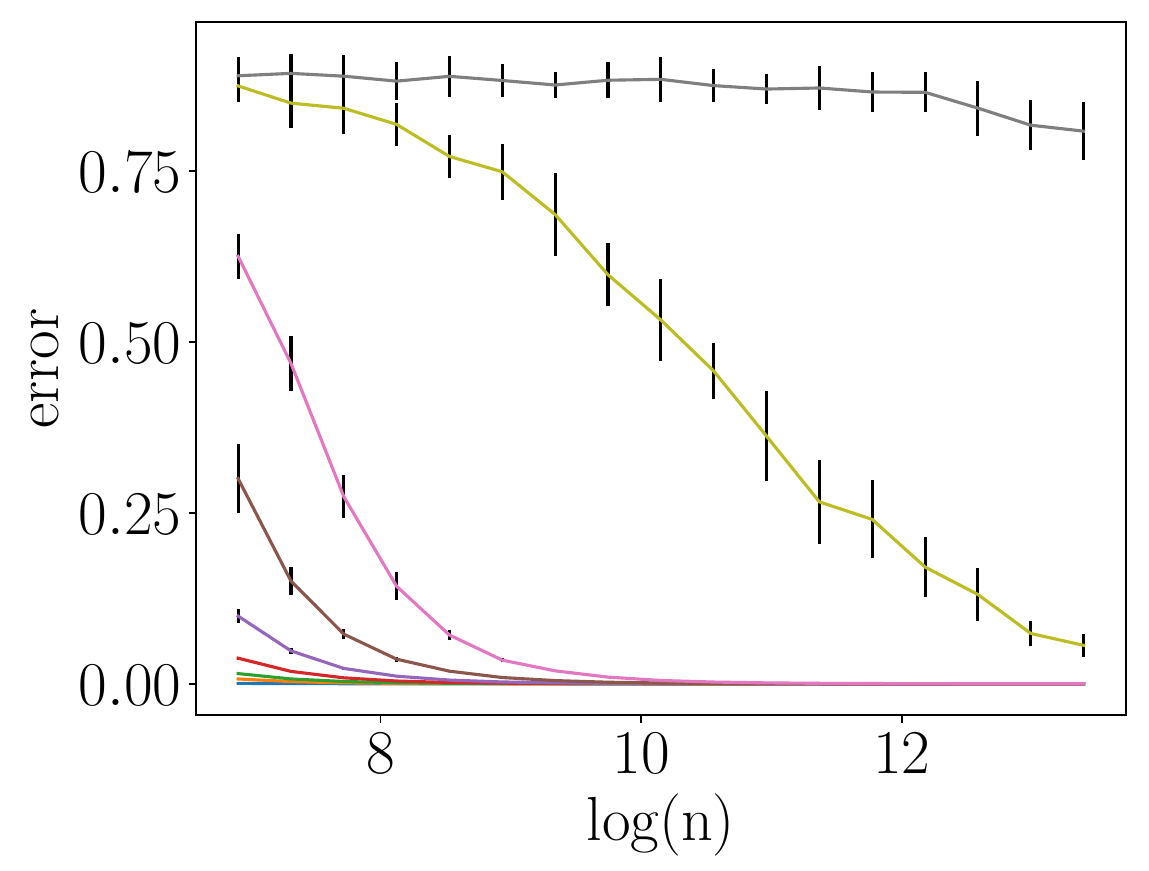}
            \includegraphics[width=0.4\linewidth]{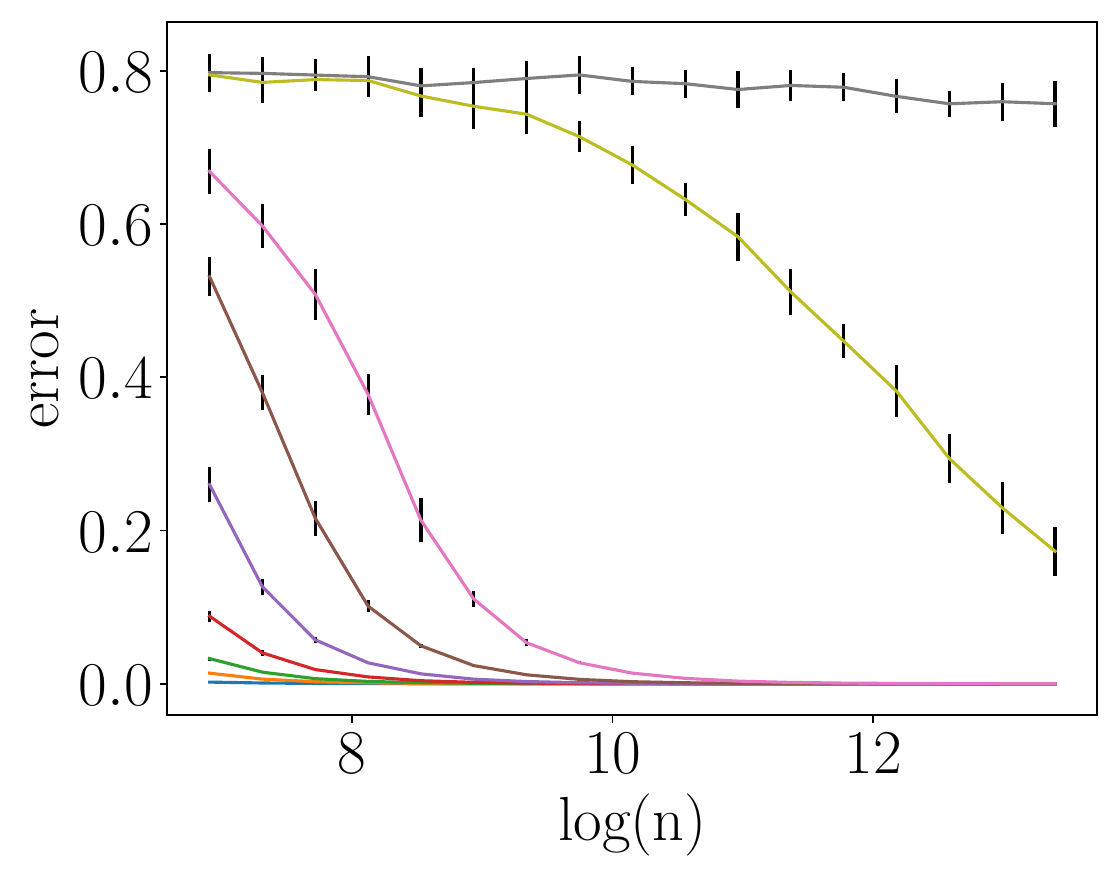}\\
            \includegraphics[width=0.6\linewidth]{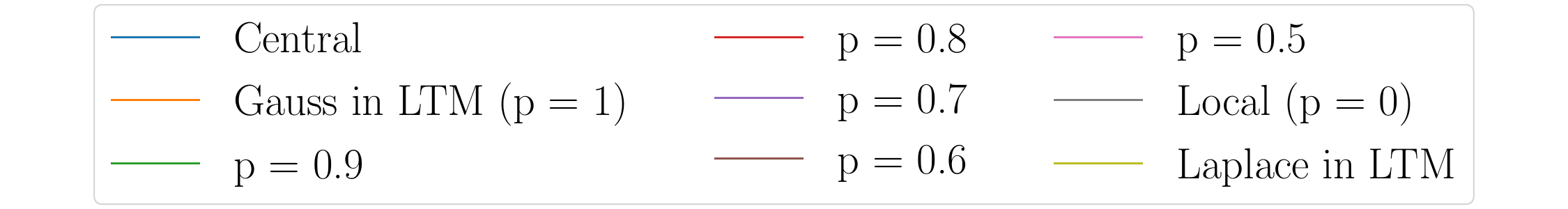}
    \end{subfigure}
    \caption{Plots depicting the asymptotic behavior of error $\psi$ for $\epsilon \in \{ 0.1,0.5 \}$ (top, bottom) and $k \in \{ 5,10 \}$ (left, right), with $d=50$. The gray line depicts the error of the local mechanism, and the orange and lime lines depict our approach using Gaussian and Laplacian noise respectively. The other lines resemble different values of $p$ when using Gaussian noise. The standard deviations are depicted by the vertical black lines and the $x$-axis is logarithmic in the number of clients $n$.}
    \label{fig:low-rank-experiment-main}
\end{figure*}
  \begin{table}[h]
    \caption{Parameters of real-world datasets.}
    \label{tab:real-world-data-low-rank-params}
    \begin{center}
        \begin{tabular}{ | l | c | c | c  | }
        \hline
        \textbf{Dataset} & $n$ & $d$ & $k$  \\ 
        \hline
        Power \cite{misc_individual_household_electric_power_consumption_235}  & $2049280$ & $6$ & $3$\\ 
        Elevation \cite{misc_3d_road_network_north_jutland_denmark_246}  & $434874$ & $2$ & $1$ \\
        Ethylene \cite{misc_gas_sensor_array_under_dynamic_gas_mixtures_322} & $4178504$ & $18$ & $5$  \\
        Songs  \cite{misc_yearpredictionmsd_203} & $515345$ & $89$ & $15$  \\
        \hline
        \end{tabular}
    \end{center}
 \end{table}

  \begin{table*}[h]
    \caption{Experimental evaluation of error $\psi$ on real-world datasets, including standard deviations. Here we are in the setting where $\epsilon = 0.05$.}
    \label{tab:real-world-data-low-rank}
    \begin{center}
        \begin{tabular}{ | l | c | c | c | c |  }
        \hline
        \textbf{Dataset}  & \textbf{Local} & \textbf{LTM Laplace} & \textbf{LTM Gauss} & \textbf{Central} \\ 
        \hline
        Power  & $ (1.206 \pm 0.651) \times 10^{-2} $ & $(4.036 \pm 2.448) \times 10^{-3}$ & $(2.108 \pm 2.400) \times 10^{-6}$ & $(1.763 \pm 0.621) \times 10^{-6}$ \\ 
        Elevation & $(1.193 \pm 0.896) \times 10^{-1}$ & $(1.520 \pm 4.484) \times 10^{-2}$ & $(4.718 \pm 4.201) \times 10^{-8}$ & $(4.465 \pm 4.836) \times 10^{-8}$ \\
        Ethylene   & $(5.877 \pm 1.256) \times 10^{-4}$ & $(5.821 \pm 1.074) \times 10^{-4}$ & $(2.420 \pm 0.486) \times 10^{-6}$ & $(7.901 \pm 1.909) \times 10^{-6}$  \\
        Songs   & $(5.617 \pm 0.283) \times 10^{-7}$ & $(5.629 \pm 0.311) \times 10^{-7}$ & $(5.556 \pm 0.347) \times 10^{-7}$ & $(5.613 \pm 0.324) \times 10^{-7}$ \\
        \hline
        \end{tabular}
    \end{center}
 \end{table*} 

In contrast to the central model of differential privacy, the performance of differentially private outputs in the local model decreases with increasing number of clients $n$.
For regression and matrix factorization problems, our theoretical results for the LTM and the Gaussian mechanism suggest that this performance decrease can be avoided without sacrificing privacy under reasonable assumptions.
It is natural to ask for which parametrizations of the Gaussian mechanism the performance remains acceptable. If more noise still yields good performance, we can tolerate more corrupted clients. 
Therefore, we evaluate empirically how the error of ridge regression and low-rank approximation develops as $n$ grows in practice and against synthetic benchmarks.

The experiments described here thus aim to demonstrate that, as the number of clients increases, the asymptotic error of differential privacy mechanisms in the LTM is between the error in the central and local models.

\mypar{Setup.}
Our Gaussian mechanism in the LTM adds noise sampled from $\mathcal{N}(0,\sigma^2)$ to every entry in input data matrix $\mdata$.
The variance $\sigma^2$ depends on privacy parameters $\epsilon$ and $\delta$ and is chosen proportional to $n^{-p}$ for $p \in [0,1]$.
This enables us to interpolate between the local model ($p = 0$) and the LTM ($p=1$). For each dataset we measure error for the central model, the Laplace mechanism and the Gaussian mechanism with $p \in \{ 1, 0.9, 0.8, 0.7, 0.6, 0.5, 0 \}$, by reporting the average error over $20$ runs of the algorithms per dataset.
The Laplace mechanism in the LTM on the other hand adds the sum of $m$ samples from $\text{Lap}(0,1/\epsilon)$ to every entry in $\sketch\mdata$ in order to investigate the $\delta = 0$ regime in the LTM.

As a baseline for low-rank approximation in the central model, we implemented MOD-SULQ \cite{DBLP:conf/nips/ChaudhuriSS12}, an approach where $\mdata^T\mdata$ is perturbed by adding Gaussian noise.
We measure the error in terms of the excess risk;
\begin{equation*}
    \psi = \frac{\lVert \mdata - \mdata\mxest\mxest^T \rVert_F^2 - \lVert \mdata - \mdata\mxopt\mxopt^T \rVert_F^2}{n},
\end{equation*}
where $\mxopt$ denotes the optimal solution and $\mxest$ denotes the solution after adding noise to the training data.
We vary the privacy parameter $k$ based on the dataset at hand.

For linear regression in the central model, we implemented the so-called Sufficient Statistics Perturbation (SSP) algorithm~\cite{DBLP:conf/icdm/VuS09}. Due to its singular dependency on a multiplicative error, we use approximation factor as an error measure:
\begin{equation*}
    \phi = \frac{\lVert \mdata \xest - \predictor \rVert_2^2 + \lambda \lVert \xest \rVert_2^2}{\lVert \mdata \xopt - \predictor \rVert_2^2 + \lambda \lVert \xopt \rVert_2^2}
\end{equation*}
where $\xopt$ denotes the optimal solution and $\xest$ denotes the solution after adding noise to the training data. 

  \begin{figure*}%[tb]
	\centering
	\begin{subfigure}%{0.49\linewidth}
        \centering
            \includegraphics[width=0.4\linewidth]{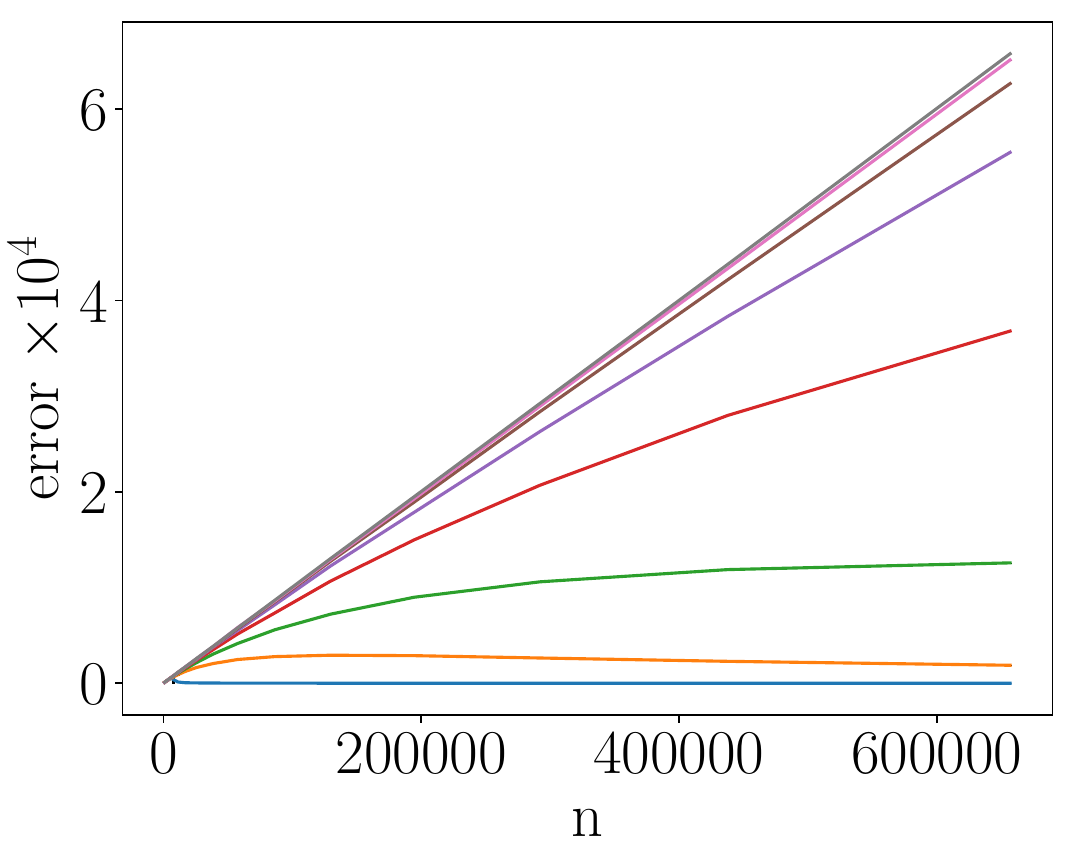}
            \includegraphics[width=0.4\linewidth]{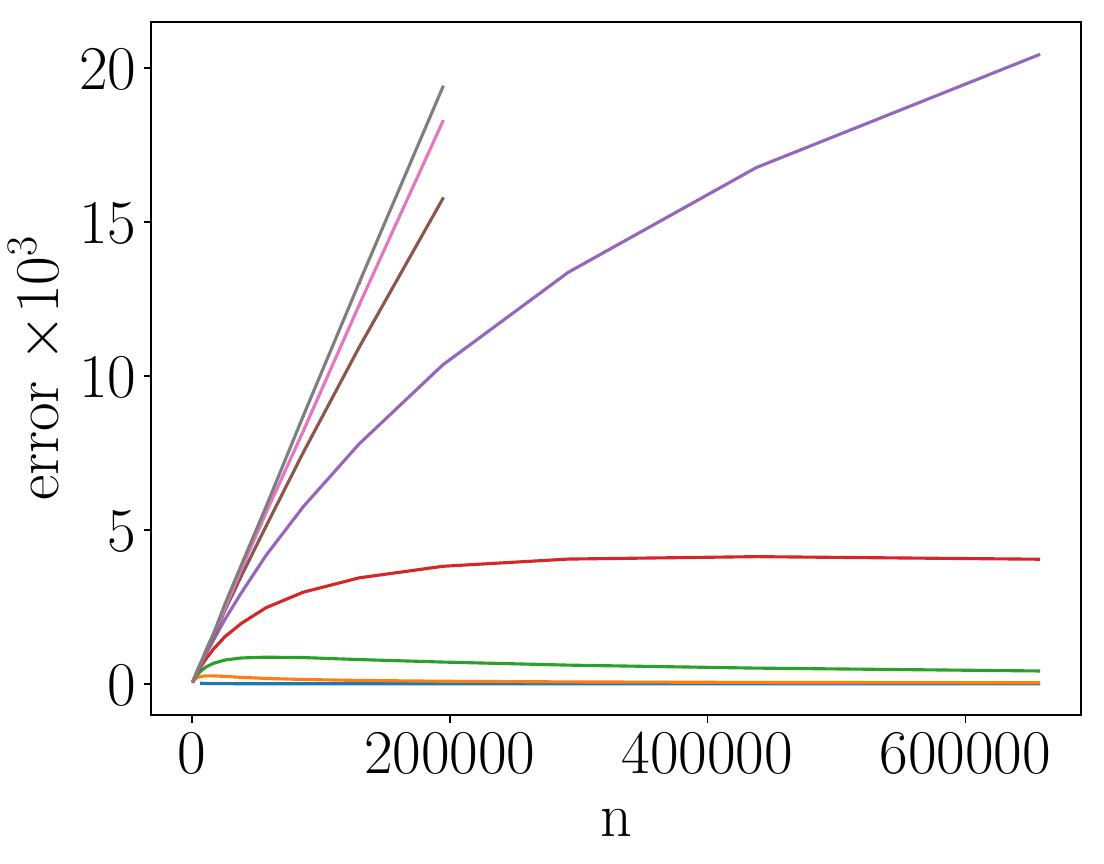}
    \end{subfigure}
    \\
    \begin{subfigure}%{0.48\linewidth}
        \centering
        \includegraphics[width=0.4\linewidth]{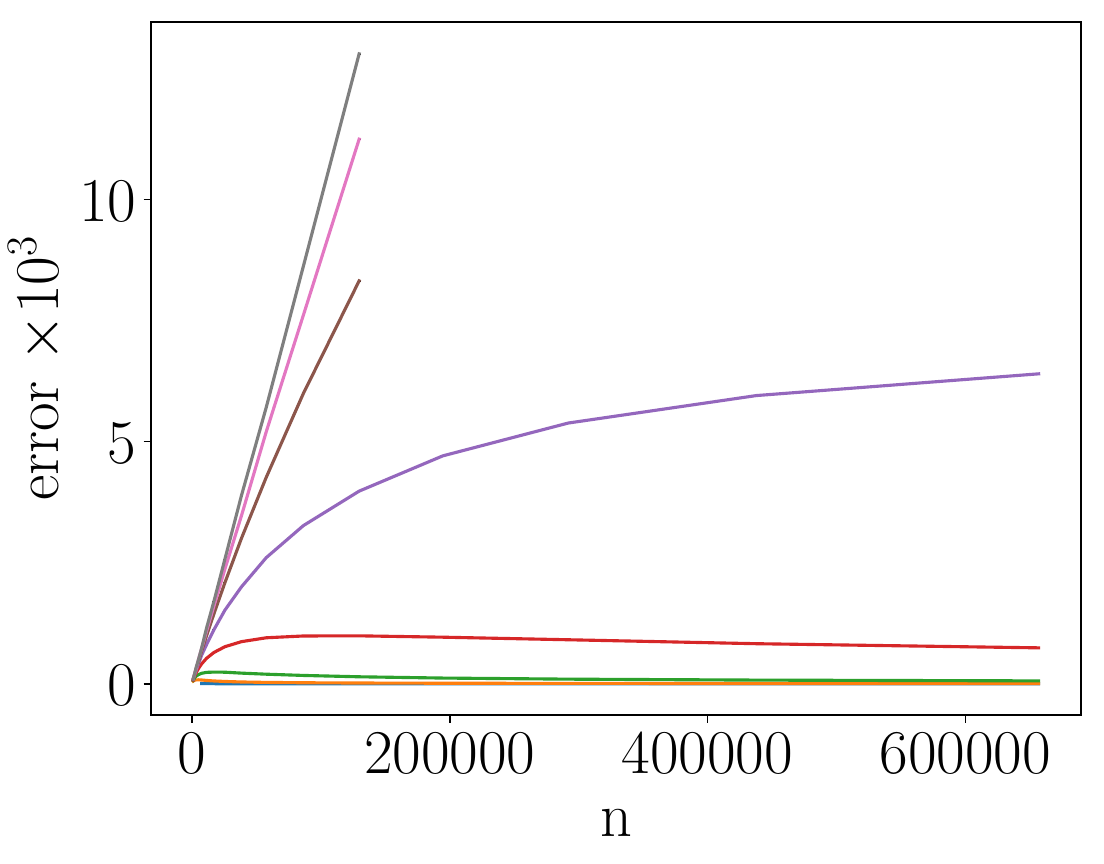}
        \includegraphics[width=0.4\linewidth]{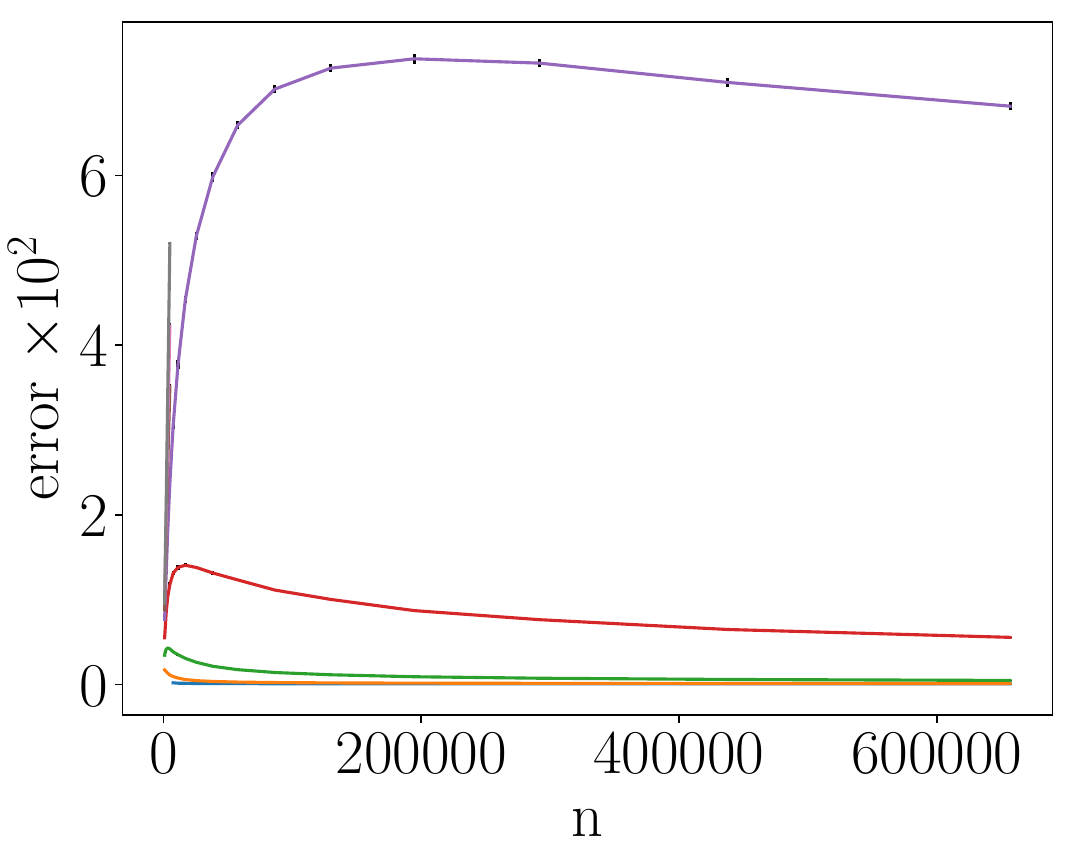} \\
        \includegraphics[width=0.6\linewidth]{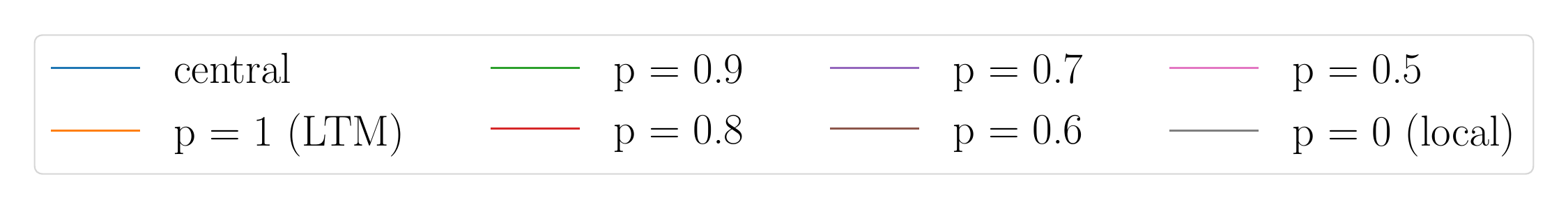}
    \end{subfigure}
    \caption{Plots depicting the asymptotic behavior of error $\phi$ for $\epsilon \in \{ 0.01,0.03,0.05,0.1 \}$ (top left, top right, bottom left, bottom right), with $d=10$, $\lambda = 10$ and $\mu^2 = n$. The grey line depicts the error of the local mechanism, the blue line does it in the central model and the orange one depicts our approach. The other lines resemble different values of $p$. In most cases the standard deviations are so small, that it is not possible to see those.}
    % \Description{}
    \label{fig:experiment}
\end{figure*}

\mypar{Datasets.}
%As it is not possible to interpolate between two different datasets in order to gain insights into asymptotic behaviour, w
We evaluate our mechanism on synthetic and real datasets, enabling  us to vary the number of sampled points $n$ and interpolating between datasets of different sizes.

The synthetic datasets for low-rank approxiamtion have a large spectral gap from the $k$th to the $(k+1)$st singular value with random bases, which emphasizes the performance difference between the various private mechanisms.
We generate synthetic data by first producing a matrix $\mdata'$ where every entry is sampled from $\mathcal{N}(0,1)$ and then changing its singular values such that there are exactly $k$ big ones and the rest are small. More specifically, we set $\mdata = \mathbf{U}'\mathbf{\Sigma}\mathbf{V}'$, where $\mathbf{U}'\mathbf{\Sigma}'\mathbf{V}'$ is the SVD of $\mdata$ and $\Sigma$ is a diagonal matrix with the first $k$ values set to $\sqrt{n / k}$ and the rest set to $1 / n$. We vary parameter $k \in \{ 5,10 \}$ for $d = 50$. 

We generate synthetic data for regression by first sampling every entry in $\mdata$ from $\mathcal{N}(0,1)$, then sampling $\xopt$ such that all entries are sampled from $\mathcal{N}(0,\mu^2)$, and finally setting $\predictor = \mdata \xopt$.
We vary parameters $d \in \{ 3,10,50 \}$, and  $\mu^2 \in \{ 1,n,n^2 \}$.

For all combinations of parameters, we generate $17$ synthetic datasets of sizes $\{ 1000^{i \cdot 1.5} | i = 0, \dots , 16 \}$. For each of those we then measure error for the central mechanism, the Laplace mechanism and the Gaussian mechanism with $p \in \{ 1, 0.9, 0.8, 0.7, 0.6, 0.5, 0 \}$, by running the algorithms $20$ times per dataset and reporting the average error.

In addition to synthetic datasets, we also evaluated our mechanism on $4$ datasets from the UC Irvine Machine Learning Repository. Table \ref{tab:real-world-data-low-rank-params} provides their number of entries $n$ and dimensionality $d$, as well as the rank $k$ we chose for low-rank approximation experiments. 

\begin{itemize}
    \item The first dataset consists of electric power consumption measurements in one household \cite{misc_individual_household_electric_power_consumption_235} and the feature we try to predict is \texttt{sub\_metering\_3}. We ignore the date and time features and the data points that had missing values. This leaves us with $6$ features (plus the one we are predicting) and $2049280$ data points. 
    \item The Elevation dataset \cite{misc_3d_road_network_north_jutland_denmark_246} consists of $434874$ open street map elevation measurements from North Jutland, Denmark. We predict the elevation from the longitude and latitude  features.
    \item The Ethylene dataset \cite{misc_gas_sensor_array_under_dynamic_gas_mixtures_322} contains recordings of sensors exposed to a mixture of gas. We trained on the part where the sensors were exposed to a mixtures of Ethylene and CO in air. The feature we are predicting is the last one \texttt{TGS2620}, which leaves us with $d = 18$ and $n = 4178504$.
    \item The Songs dataset \cite{misc_yearpredictionmsd_203}, consists of $89$ audio features that are meant to predict the release year of a song.
\end{itemize}

\mypar{Low-Rank Approximation.}
Figure \ref{fig:low-rank-experiment-main} shows our results for two privacy regimes ($\epsilon \in \{ 0.1, 0.5 \}$) with synthetic data and Table \ref{tab:real-world-data-low-rank} shows the error for $\epsilon = 0.05$ with real-world datasets.
In all chosen parameter settings, we observe that as $n$ grows, the error in the LTM asymptotically approaches the error in the central model, both on real and synthetic data sets.
On real-world datasets, our Gaussian approach performs significantly better than the Gaussian mechanism in the local model and is very close to the central MOD-SULQ \cite{DBLP:conf/nips/ChaudhuriSS12} mechanism. 
Both with synthetic and real-world datasets, our Laplace approach performs better than the local model, though not as significantly as the Gaussian approach.
However, this is the price of achieving pure $\epsilon$-differential privacy with $\delta = 0$.

\begin{figure*}
    \centering
    \includegraphics[width=0.9\linewidth]{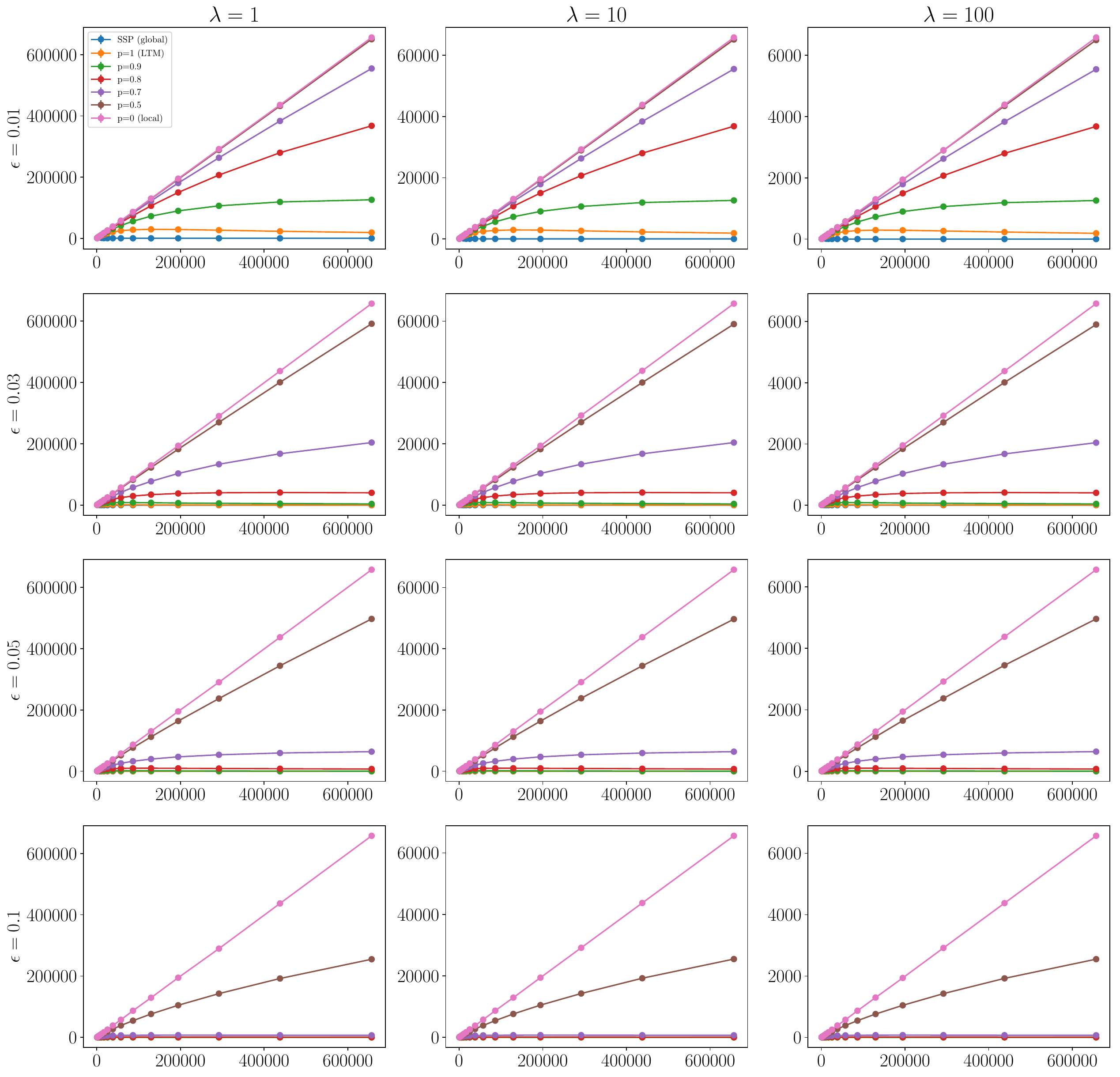}
    \caption{Plots depicting the error $\phi$ as $n$ increases, for $\epsilon \in \{ 0.01,0.03,0.05,0.1 \}$ (top to bottom) and $\lambda \in \{ 1,10,100 \}$ (left to right). The middle row depicts the same choice of parameters as Figure \ref{fig:experiment}.}
    % \Description{}
    \label{fig:experiment-appendix}
\end{figure*}
 \begin{table}%[h]
    \caption{Experimental evaluation of error $\phi$ on real-world datasets, including standard deviations. Here we are in the setting where $\epsilon = 0.03$ and $\lambda = 10$.}
    \label{tab:real-world-data}
    \begin{center}
        \begin{tabular}{ | l | c | c | c | }
        \hline
        \textbf{Dataset} & Local & Our & Central \\ 
        \hline
        Power & $2.364 \pm .007$ & $1.055 \pm .000$ & $1.001 \pm .001$ \\  
        Elevation & $1.939 \pm .008$ & $1.000 \pm .000$ & $1.000 \pm .000$ \\
        Ethylene & $3.125 \pm .002$ & $1.000 \pm .000$ & $1.000 \pm .000$ \\
        Songs & $20.64 \pm .016$ & $1.000 \pm .000$ & $1.000 \pm .000$ \\
        \hline
        \end{tabular}
    \end{center}
 \end{table}

\mypar{Ridge Regression.}
Figure~\ref{fig:experiment} provides our experimental results based on synthetic data, 
%(see also Figure~\ref{fig:experiment-big} in Appendix~\ref{experiments_appendix})
and Table~\ref{tab:real-world-data} shows the error resulting from real-world datasets. As expected, our approach performs asymptotically better than the Gaussian mechanism in the local model, but worse than SSP in the central model. We also found that as $n$ increases, the error of our approach asymptotically approaches the error in the central model. For $p = 0.9, 0.8, 0.7, 0.6$, the error eventually decreases for a sufficiently large $n$ for the privacy regimes we considered ($\epsilon \in \{ 0.01,0.03,0.05,0.1 \}$). 

For $p = 0.5$ though we saw a significant jump towards the local model in terms of asymptotic behaviour. In all the settings we considered, $p=0.5$ showed an asymptotic increase of the error. We tested on synthetic datasets of up to $n = 40$ million, and the error also increases in this regime. 

Table~\ref{tab:real-world-data} shows the results we got by applying the local, our mechanism and SSP on real-world data. They include standard deviations, though for many of the results those are so small that they appear as $0$ in the table. We found, that our mechanism performs significantly better than the Gaussian mechanism for local privacy. In most it even performs so well that there is basically no difference to the central model.

\paragraph{Varying $\lambda$}
In Figure \ref{fig:experiment-appendix}, we provide plots for the same setting of parameters as on Figure \ref{fig:experiment} ($d = 10$ and $\mu^2 = n$). Though, we also vary $\lambda \in \{ 1,10,100 \}$ here.
The experiments show that varying $\lambda$ does not change the asymptotic behavior of any of the mechanisms we investigated, in the sense that a variance proportional to $n^{-0.6}$ or lower eventually had decreasing error, while a variance proportional to $n^{0.5}$ produced increasing errors. The increase of $\lambda$ does however produce a significantly lower error in all settings (note the $y$-axes on the figure). Our theoretical results likewise predict this behaviour, as while the error bounds of Theorem~\ref{thrm:linreg} improve with increasing $\lambda$, they do not affect the dependency on $n$. Thus, we view the experiments as a confirmation that the theoretical bounds, while potentially improvable, express the correct asymptotic relationship between parameters and approximation bounds.

\begin{acks}
    The research described in this paper has received funding from: the European Research Council (ERC) under the European Unions's Horizon 2020 research and innovation programme under grant agreement No 803096 (SPEC); the Danish Independent Research Council under Grant-ID DFF-2064-00016B (YOSO); DFF-3103-00077B (CryptoDigi) and the Danish Independent Research Council under Grant-ID 1051-001068 (ACBD).
\end{acks}

\bibliographystyle{ACM-Reference-Format}
\bibliography{bibliography}

\end{document}